\documentclass[twoside,11pt]{article}
\usepackage[margin=1in]{geometry}

\usepackage{amsfonts,amssymb,color,amsthm,amsmath}

\usepackage[round,sort]{natbib}
\usepackage{hyperref}
\hypersetup{colorlinks=true,linkcolor=red,citecolor=blue}

\usepackage{authblk}

\usepackage{ctable} 
\usepackage{multirow}

\usepackage{graphicx,float}

\usepackage{tikz}
\usepackage{boxedminipage}

\usepackage[ruled,linesnumbered]{algorithm2e}

\theoremstyle{plain}
\newtheorem{theorem}{Theorem}
\newtheorem{lemma}{Lemma}[]
\newtheorem{proposition}{Proposition}[] 
\newtheorem{corollary}{Corollary}[]

\theoremstyle{definition}
\newtheorem{definition}{Definition}[]

\theoremstyle{remark}
\newtheorem{remark}{Remark}[]

\newcommand{\argmin}{\mathop{\text{argmin}}}
\newcommand{\bs}{\boldsymbol}

\newcommand{\vertiii}[1]{{\left\vert\kern-0.25ex\left\vert\kern-0.25ex\left\vert #1 
    \right\vert\kern-0.25ex\right\vert\kern-0.25ex\right\vert}}

\newcommand{\fnorm}[1]{\left\lVert#1\right\rVert_F}

\newcommand{\dom}{\mbox{dom}}

\renewenvironment{proof}{{\noindent\bfseries\em Proof.}}{$\hfill\blacksquare$\newline}
\newcommand*\samethanks[1][\value{footnote}]{\footnotemark[#1]}

\begin{document}

\title{\bf Penalized Maximum Likelihood Estimation of \\ Multi-layered Gaussian Graphical Models}

\author{Jiahe Lin \thanks{Equal Contribution.}}
\affil{{\em University of Michigan}. jiahelin@umich.edu}

\author{Sumanta Basu \samethanks}
\affil{{\em University of California, Berkeley}. sumbose@berkeley.edu}

\author{Moulinath Banerjee}
\affil{{\em University of Michigan}. moulib@umich.edu}

\author{George Michailidis \thanks{Corresponding Author.}}
\affil{{\em University of Florida}. gmichail@ufl.edu}

\date{\vspace{-5ex}}
\maketitle
\begin{abstract} 
Analyzing multi-layered graphical models provides insight into understanding the conditional relationships among nodes within layers after adjusting for and quantifying the effects of nodes from other layers. We obtain the penalized maximum likelihood estimator for Gaussian multi-layered graphical models, based on a computational approach involving screening of variables, iterative estimation of the directed edges between layers and undirected edges within layers and a final refitting and stability selection step that provides improved performance in finite sample settings. We establish the consistency of the estimator in a high-dimensional setting. To obtain this result, we develop a strategy that leverages the biconvexity of the likelihood function to ensure convergence of the developed iterative algorithm to a stationary point, as well as careful uniform error control of the estimates over iterations.  The performance of the maximum likelihood estimator is illustrated on synthetic data. 
\end{abstract}


\section{Introduction}

The estimation of directed and undirected graphs from high-dimensional data has received a lot of attention in the machine
learning and statistics literature \cite[e.g., see][and references therein]{buhlmann2011statistics}, due to
their importance in diverse applications including understanding of biological processes and disease mechanisms, financial
systems stability and social interactions, just to name a few \citep{sachs2005causal,wang2007local,sobel2000causal}. In the case of undirected graphs, the edges capture
conditional dependence relationships between the nodes, while for directed graphs they are used to model causal relationships \citep{buhlmann2011statistics}. 

However, in a number of applications the nodes can be {\em naturally partitioned} into sets that exhibit interactions both between
them and amongst them. As an example, consider an experiment where one has collected data for both genes and metabolites for
the same set of patient specimens. In this case, we have three types of interactions between genes and metabolites: regulatory interactions between the two of them and co-regulation within the gene and within the metabolic compartments. The latter two types of relationships can be expressed through undirected graphs within the sets of genes and metabolites, respectively, while
the regulation of metabolites by genes corresponds to directed edges. Note that in principle there are feedback mechanisms from the
metabolic compartment to the gene one, but these are difficult to detect and adequately estimate in the absence of carefully collected time course data. Another example comes from the area of financial economics, where one collects data on returns of financial assets (e.g. stocks, bonds) and also on key macroeconomic indicators (e.g. interest rate, prices indices, various
measures of money supply and various unemployment indices). Once again, over short time periods there is influence from
the economic variables to the returns (directed edges), while there are co-dependence relationships between the asset returns and the macroeconomic variables, respectively, that can be modeled as undirected edges.

Technically, such {\em layered} network structures correspond to multipartite graphs that possess undirected edges and exhibit a 
directed acyclic graph structure between the layers, as depicted in Figure~\ref{fig:diagram}, where we use directed solid edges to denote the dependencies across layers and dashed undirected edges to denote within-layer conditional depedencies. 
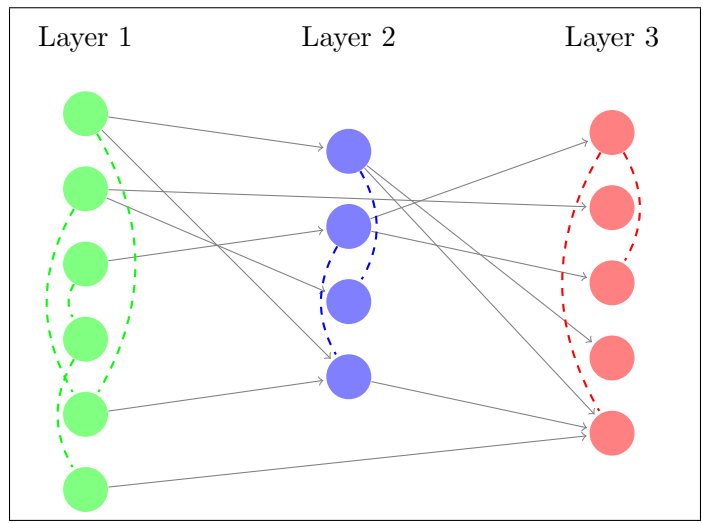
\begin{figure}[h]
\centering
\caption{Diagram for a three-layered network}\vspace*{2mm}\label{fig:diagram}
\begin{boxedminipage}{9.2cm}
\def\layersep{3.5cm}
\begin{tikzpicture}[shorten >=1pt,->,draw=black!50, node distance=\layersep]
    \tikzstyle{every pin edge}=[<-,shorten <=1pt]
    \tikzstyle{neuron}=[circle,fill=black!25,minimum size=17pt,inner sep=0pt]
    \tikzstyle{input neuron}=[neuron, fill=green!50];
    \tikzstyle{output neuron}=[neuron, fill=red!50];
    \tikzstyle{hidden neuron}=[neuron, fill=blue!50];
    \tikzstyle{annot} = [text width=4em, text centered]

    \foreach \name / \y in {1,...,6}
        \node[input neuron] (I-\name) at (0,-\y) {};

    \foreach \name / \y in {1,...,4}
        \path[yshift=-0.5cm]
            node[hidden neuron] (H-\name) at (\layersep,-\y cm) {};

    \foreach \name / \y in {1,...,5} 
	   \path[yshift=-0.25cm]
		   node[output neuron] (O-\name) at (2*\layersep, -\y cm) {};
	    
    \path (I-1) edge (H-1); \path (I-1) edge (H-4);
    \path (I-2) edge (H-3);
    \path (I-5) edge (H-4);
    \path (I-3) edge (H-2);

    \path (H-2) edge (O-1); \path (H-2) edge (O-3);
    \path (H-1) edge (O-4); \path (H-1) edge (O-5);
    \path (H-4) edge (O-5);
    
    \path (I-2) edge (O-2); \path (I-6) edge (O-5);
    
    \path[-,dashed, green,thick,bend left] (I-1) edge (I-5);
    \path[-,dashed, green,thick,bend right](I-2) edge (I-5);
	\path[-,dashed, green,thick,bend right](I-3) edge (I-4);
    \path[-,dashed, green,thick,bend right](I-4) edge (I-6);
    
    \path[-,dashed, blue,thick,bend left] (H-1) edge (H-3);
    \path[-,dashed, blue,thick,bend right] (H-2) edge (H-4);
     
    \path[-,dashed, red,thick,bend left] (O-1) edge (O-3);
    \path[-,dashed, red,thick,bend right] (O-1) edge (O-5);
    
    \node[annot,above of=H-1, node distance=1.5cm] (hl) {Layer 2};
    \node[annot,left of=hl] {Layer 1};
    \node[annot,right of=hl] {Layer 3};
\end{tikzpicture}
\end{boxedminipage}
\end{figure} Selected properties of such so-called {\em chain graphs} have been studied in the work of \citet{drton2008sinful}, with an emphasis on two alternative Markov properties including the LWF Markov property \citep{lauritzen1989graphical,frydenberg1990chain} and the AMP Markov property \citep{andersson2001alternative}. 

While layered networks being interesting from a theoretical perspective and having significant scope for applications, 
their estimation has received little attention in the literature. 
Note that for a 2-layered structure, the directed edges can be obtained through a multivariate regression procedure, while
the undirected edges in both layers through existing procedures for graphical models (for more technical details see Section~\ref{sec:estimation}). This is the strategy leveraged in the work of \citet{rothman2010sparse}, where for a 2-layered network structure proposed a multivariate regression with covariance estimation (MRCE) method for estimating the undirected edges in the second layer and the directed edges between them. A coordinate
descent algorithm was introduced to estimate the directed edges, while the popular glasso estimator \citep{friedman2008sparse} was used for the undirected edges. However, this method does not scale well according to the simulation results presented and no theoretical properties of the
estimates were provided. In follow-up work, \citet{lee2012simultaneous} proposed the Plug-in Joint Weighted Lasso (PWL) and the Plug-in Joint Graphical Weighted Lasso (PWGL) estimator for estimating the same 2-layered structure, where they use a weighted version of the algorithm in \citet{rothman2010sparse} and also provide theoretical results for the low dimensional setting, where the number of samples exceeds the number of potential directed and undirected edges to be estimated. Finally, \citet{cai2012covariate} proposed a method for estimating the same 2-layered structure and provided corresponding theoretical results in the high dimensional setting. The Dantzig-type estimator \citep{candes2007dantzig} was used for the regression coefficients and the corresponding residuals were used as surrogates, for obtaining the precision matrix through the CLIME estimator \citep{cai2011constrained}. While the above work assumed a Gaussian distribution for the data, in
more recent work by \citet{yang2014mixed}, the authors constructed the model under a general {\em mixed graphical model} framework, which allows each node-conditional distribution to belong to a potentially different univariate exponential family. In particular, with an underlying {\em mixed MRF} graph structure, instead of maximizing the joint likelihood, the authors proposed to estimate the homogeneous and heterogenous neighborhood for each node (which corresponds to undirected  and directed edges respectively, if put in the layered-network setting) by obtaining the $\ell_1$ regularized $M$-estimator of the node-conditional distribution parameters, using traditional approaches \citep[e.g.][]{meinshausen2006high} for neighborhood estimation. However, if we consider the overall error incurred by the neighborhood selection procedure of each individual node, the error bound becomes not tight due to the union bound operation used to obtain it.

In this work, we obtain the regularized maximum likelihood estimator under a sparsity assumption on both directed and undirected parameters for multi-layered Gaussian graphical models and establish its consistency properties in a high-dimensional setting. 
As discussed in Section~\ref{sec:Theory}, the problem is {\em not jointly convex} on the parameters, but convex on selected subsets of them. Further, it turns out that the problem is {\em biconvex} if
we consider a recursive multi-stage estimation approach that at each stage involves only regression parameters (directed edges) from preceeding layers and precision matrix parameters (undirected edges) for the {\em last layer considered} in that stage. Hence, we decompose the multi-layer network structure estimation into a sequence of 2-layer problems that allows us to establish the desired results.
Leveraging the biconvexity of the 2-layer problem, we establish the convergence of the iterates to the maximum-likelihood estimator, which under certain regularity conditions is arbitrarily close to the true parameters. The theoretical guarantees provided require a {\em uniform control} of the precision of the
regression and precision matrix parameters, which poses a number of theoretical challenges resolved in Section ~\ref{sec:Theory}.

In summary, despite the lack of overall convexity, we are able
to provide theoretical guarantees for the MLE in a high dimensional setting. 
We believe that the proposed strategy is generally applicable
to other non-convex statistical estimation problems that can be decomposed to two biconvex problems. Further, to enhance the numerical performance of the MLE in finite (and small) sample settings, we introduce a screening step that selects active nodes for the iterative algorithm used and that leverages recent developments in the high-dimensional regression literature \citep[e.g.,][]{van2014asymptotically,javanmard2014confidence,zhang2014confidence}. We also post-process the final MLE estimate through a stability selection procedure. As mentioned above, the screening and stability selection steps are beneficial to the performance of the MLE in finite samples and hence recommended for similarly structured problems.

The remainder of the paper is organized as follows. In Section \ref{sec:methodology}, we introduce the proposed methodology, with an emphasis on how the multi-layered network estimation problem is decomposed into a sequence of  two-layered network estimation problem(s).  In Section \ref{sec:Theory}, we provide theoretical guarantees for the estimation procedure posited. In particular, we show consistency of the estimates and convergence of the algorithm, under a number of common assumptions in high-dimensional settings. In Section \ref{sec: Implementation}, we show the performance of the proposed algorithm with simulation results under different simulation settings, and introduce serveral accerleration techniques which speed up the convergence of the algorithm and reduce the computing time in practical settings.

%
%
\section{Problem Formulation.}\label{sec:methodology}


Consider an $M$-layered Gaussian graphical model. Suppose there are $p_m$ nodes in Layer $m$, denoted by 
\begin{equation*}
\bs{X}^m=(X^m_1,\cdots,X^m_{p_m})', \quad \text{for }m=1,\cdots,M.
\end{equation*}
The structure of the model is given as follows: 
\begin{itemize}
\item[--] Layer 1. $\bs{X}^1=(X^1_1,\cdots,X^1_{p_1})'\sim \mathcal{N}(0,\Sigma^1)$.
\item[--] Layer 2. For $j=1,\cdots,p_2$: $X^2_j=(B^{12}_j)'\bs{X}^1 + \epsilon_j^2$, with $B^{12}_j\in\mathbb{R}^{p_1}$, and $\bs{\epsilon}^2 = (\epsilon^2_1,\cdots,\epsilon^2_{p_2})'\sim \mathcal{N}(0,\Sigma^2)$.
\item[ ] $\vdots$
\item[--] Layer $M$. For $j=1,2,\cdots,p_M$: 
\begin{equation*}
X^M_j = \sum_{m=1}^{M-1} \{(B^{mM}_j)'\bs{X}^m \}+ \epsilon_j^M, \quad \text{where }B^{mM}_j\in\mathbb{R}^{p_m}~~\text{for }m=1,\cdots,M-1,
\end{equation*}
and $\bs{\epsilon}^M=(\epsilon^M_1,\cdots,\epsilon^M_{p_M})'\sim\mathcal{N}(0,\Sigma^M)$.
\end{itemize}
The parameters of interest are {\em all directed edges} that encode the dependencies across layers, that is:
\begin{equation*}
B^{st} := \begin{bmatrix}B^{st}_1 & \cdots & B^{st}_{p_t}\end{bmatrix}, \quad \text{for } 1\leq s < t \leq M,
\end{equation*}
and {\em all undirected edges} that encode the conditional dependencies within layers after adjusting for the effects from directed edges, that is:
\begin{equation*}
\Theta^m := (\Sigma^m)^{-1}, \quad \text{for }m=1,\cdots,M.
\end{equation*}
It is assumed that $B^{st}$ and $\Theta^m$ are {\em sparse} for all $1,\dotsc, M$ and $1\leq s < t\leq M$.

Given centered data for all $M$ layers, denoted by $X^m=[X_1^m,\cdots,X^m_{p_m}]\in\mathbb{R}^{n\times p_m}$ for all $m=1,\cdots,M$, we aim to obtain the MLE for all $B^{st},1\leq s<t\leq M$ and all $\Theta^m,m=1,\cdots,M$ parameters. Henceforth, we use $\bs{X}^m$ to denote random vectors, and $X_j^m$ to denote the $j$th column in the data matrix $X_{n \times p_m}^m$ whenever there is no ambiguity. 

\medskip
Through Markov factorization \citep{lauritzen1996graphical}, the full log-likelihood function can be decomposed as:
\begin{equation*}
\scriptsize
\begin{split}
\ell(X^m; B^{st},\Theta^m,1\leq s<t\leq M,1\leq m\leq M)& =\ell(X^M|X^{M-1},\cdots,X^1;B^{1M},\cdots,B^{M-1,M},\Theta^M) \\ 
&\quad  + \ell(X^{M-1}|X^{M-2},\cdots,X^1;B^{1M-1},\cdots,B^{M-2,M-1},\Theta^{M-1})\\
& \quad + \cdots + \ell(X^2|X^1; B^{12},\Theta^2)+ \ell(X^1; \Theta^1)\\
& = \ell(X^1; \Theta^1) + \sum\nolimits_{m=2}^M \ell(X^m|X^1,\cdots,X^{m-1}; B^{1m},\cdots,B^{m-1,m},\Theta^m). 
\end{split}
\end{equation*}
Note that the summands share no common parameters, which enables us to maximize the likelihood with respect to individual parameters in the $M$ terms separately. More importantly, by conditioning Layer $m$ nodes on nodes in its previous $(m-1)$ layers, we can treat Layer $m$ nodes as the``response" layer, and all nodes in the previous $(m-1)$ layer combined as a super ``parent" layer. If we ignore the structure within the bottom layer ($X^1$) for the moment, the $M$-layered network can be viewed as $(M-1)$ two-layered networks, each comprising a response layer and a parent layer. Thus, the network structure in Figure~\ref{fig:diagram} can be viewed as a 2 two-layered network: for the first network, Layer 3 is the response layer, while Layers 1 and 2 combined form the ``parent" layer; for the second network, Layer 2 is the response layer, and Layer 1 is the ``parent" layer. Therefore, the problem for estimating all $\binom{M}{2}$ coefficient matrices and $M$ precision matrices can be translated into estimating $(M-1)$ two-layered network structures with directed edges from the parent layer to the response layer, and undirected edges within the response layer, and finally estimating the undirected edges within the bottom layer separately. 

Since all estimation problems boil down to estimating the structure of a 2-layered network, we focus the technical discussion on introducing our proposed methodology for a 2-layered network setting\footnote{In Appendix~\ref{appendix:example}, we give a detail example on how our proposed method works under a 3-layered network setting.}. The theoretical results obtained extend in a straightforward manner to an $M$-layered Gaussian graphical model.

\noindent
\begin{remark}
For the $M$-layer network structure, we impose certain identifiability-type condition on the largest ``parent" layer (encompassing $M-1$ layers), so that the directed edges of the entire network are estimable. The imposed condition translates into a minimum eigenvalue-type condition on the population precision matrix within layers, and conditions on the magnitude of dependencies across layers. Intuitively, consider a three-layered network: if $\bs{X}^1$ and $\bs{X}^2$ are highly correlated, then the proposed (as well as any other) method will exhibit difficulties in distinguishing the effect of $\bs{X}^1$ on $\bs{X}^3$ from that of $\bs{X}^2$ on $\bs{X}^3$. The (group) identifiability-type condition is thus imposed to obviate such circumstances. An in-depth discussion on this issue is provided in Section~\ref{sec:identifiability}.
\end{remark}

\subsection{\normalsize A Two-layered Network Set-up.}\label{sec:set-up}

Consider a two-layered Gaussian graphical model with $p_1$ nodes in the first layer, denoted by $\bs{X}=(X_1,\cdots,X_{p_1})'$, and $p_2$ nodes in the second layers, denoted by $\bs{Y}=(Y_1,\cdots,Y_{p_2})'$. The model is defined as follows: 
\begin{itemize}
\setlength\itemsep{1pt}
\renewcommand\labelitemi{--}
\item $\bs{X}=(X_1,\cdots,X_{p_1})'\sim \mathcal{N}(0,\Sigma_X)$.
\item For $j=1,2,\cdots,p_2$: $Y_j = B_j' \bs{X} + \epsilon_j$, $B_j\in\mathbb{R}^{p_1}$ and $\bs{\epsilon}=(\epsilon_1,\cdots,\epsilon_{p_2})^\top\sim\mathcal{N}(0,\Sigma_\epsilon)$.
\end{itemize}
The parameters of interest are: $\Theta_X:=\Sigma_X^{-1},\Theta_\epsilon:=\Sigma_\epsilon^{-1}$ and $B=[B_1,\cdots,B_{p_2}]$. As with most estimation problems in the high dimensional setting, we assume these parameters to be sparse. 

Now given data $X=[X_1,\cdots,X_{p_1}]\in\mathbb{R}^{n\times p_1}$ and $Y=[Y_1,\cdots,Y_{p_2}]\in\mathbb{R}^{n\times p_2}$, both centered, we would like to use the penalized maximum likelihood approach to obtain estimates for $\Theta_X$, $\Theta_\epsilon$ and $B$. Throughout this paper, we use $X$, $Y$ and $E$ to denote the size-$n$ realizations of the random vectors $\bs{X}$, $\bs{Y}$ and $\bs{\epsilon}$, respectively. Also, with a slight abuse of notation, we use $X_i,i=1,2,\cdots,p_1$ and $Y_{j},j=1,2,\cdots,p_2$ to denote the columns of the data matrix $X$ and $Y$, respectively, whenever there is no ambiguity. \\

The full log-likelihood can be written as 
\begin{equation}\label{eqn:opt}
\ell(X,Y;B,\Theta_\epsilon,\Theta_X) = \ell(Y|X;\Theta_\epsilon,B) + \ell(X;\Theta_X)
\end{equation}
Note that the first term only involves $\Theta_\epsilon$ and $B$, and the second term only involves $\Theta_X$. Hence, (\ref{eqn:opt}) can be maximized by maximizing $\ell(Y|X)$ w.r.t. $(\Theta_\epsilon,B)$, and maximizing $\ell(X)$ w.r.t. $\Theta_X$, respectively. $\widehat{\Theta}_X$ can be obtained using traditional methods for estimating undirected graphs, e.g., the Graphical Lasso \citep{friedman2008sparse} or the Nodewise Regression prcoedure \citep{meinshausen2006high}. Therefore, the rest of this paper will mainly focus on obtaining estimates for $\Theta_\epsilon$ and $B$. In the next subsection, we introduce our estimation procedure for obtaining the MLE for $\Theta_\epsilon$ and 
$B$.

\noindent
\begin{remark}
Our proposed method is targeted towards maximizing $\ell(Y|X;\Theta_\epsilon,B)$ (with proper penalization) in (\ref{eqn:opt}) only, which gives the estimates for across-layers dependencies between the response layer and the parent layer, as well as estimates for the conditional dependencies within the response layer each time we solve a 2-layered network estimation problem. For an $M$-layered estimation problem, the maximization regarding $\ell(X;\Theta_X)$ occurs only when we are estimating the within-layer conditional dependencies for the bottom layer.  
\end{remark}
%
%

\subsection{\normalsize Estimation Algorithm.}\label{sec:estimation}

The conditional likelihood for response $Y$ given $X$ can be written as:
\begin{eqnarray*}
L(Y|X) & = (\frac{1}{\sqrt{2\pi}})^{n{p_2}} |\Sigma_\epsilon\otimes I_n|^{-1/2}\exp\begin{Bmatrix}-\frac{1}{2} (\mathcal{Y}-\mathcal{X}\bs{\beta})^\top(\Sigma_\epsilon\otimes I_n)^{-1}(\mathcal{Y}-\mathcal{X}\bs{\beta})  \end{Bmatrix},
\end{eqnarray*}
where $\mathcal{Y} = vec(Y_1,\cdots,Y_{p_2})$, $\mathcal{X}=I_{p_2}\otimes X$ and $\beta=vec(B_1,\cdots,B_{p_2})$. After writing out the Kronecker product, the log-likelihood can be written as:
\begin{equation*}
\ell (Y|X) = \text{constant}+\frac{n}{2}\log\det\Theta_\epsilon-\frac{1}{2}\sum_{j=1}^{p_2}\sum_{i=1}^{p_2}\sigma^{ij}_\epsilon(Y_i-XB_i)^\top(Y_j-XB_j).
\end{equation*}
Here, $\sigma^{ij}_{\epsilon}$ denotes the $ij$-th entry of $\Theta_\epsilon$. With $\ell_1$ penalization which induces sparsity, the optimization problem can be formulated as:
\begin{equation}\label{eqn:obj}
\min\limits_{\substack{B\in\mathbb{R}^{p_1\times p_2} \\ \Theta_\epsilon\in \mathbb{S}_{++}^{p_2\times p_2}}} \left\{
\frac{1}{n}\sum_{j=1}^{p_2}\sum_{i=1}^{p_2} \sigma_\epsilon^{ij}(Y_i-XB_i)^\top(Y_j-XB_j)-\log\det \Theta_\epsilon + \lambda_n\sum_{j=1}^{p_2} \|B_j\|_1 + \rho_n \|\Theta_\epsilon\|_{1,\text{off}} \right\},
\end{equation}
and the first term in (\ref{eqn:obj}) can be equivalently written as:
\begin{equation*}
\text{tr}\begin{Bmatrix}\frac{1}{n}\begin{bmatrix}
(Y_1-XB_1)^\top \\  \vdots \\ (Y_{p_2}-XB_{p_2})^\top  
\end{bmatrix} \begin{bmatrix}
(Y_1-XB_1) & \cdots & (Y_{p_2}-XB_{p_2})
\end{bmatrix}\Theta_\epsilon \end{Bmatrix}:=\text{tr}(S\Theta_\epsilon).
\end{equation*}
where $S$ is defined as the sample covariance matrix of $E\equiv Y-XB$. This gives rise to the following optimization problem:
\begin{equation}\label{eqn:obj2}
\min\limits_{\substack{B\in\mathbb{R}^{p_1\times p_2} \\\Theta_\epsilon\in \mathbb{S}_{++}^{p_2\times p_2}}}  \left\{ \text{tr}(S\Theta_\epsilon) - \log\det\Theta_\epsilon +  \lambda_n\sum_{j=1}^{p_2} \|B_j\|_1 + \rho_n \|\Theta_\epsilon\|_{1,\text{off}} \right\}\equiv f(B,\Theta_\epsilon),
\end{equation}
where $\|\Theta\|_{1,\text{off}}$ is the absulote sum of the off-diagonal entries in $\Theta$, $\lambda_n$ and $\rho_n$ are both positive tuning parameters. This penalized log-likelihood corresponds to the objective function initially proposed in \citet{rothman2010sparse}, and has also been examined in \citet{lee2012simultaneous}. \\

Note that the objective function (\ref{eqn:obj2}) is {\em not jointly convex} in $(B,\Theta_\epsilon)$, but only 
convex in $B$ for fixed $\Theta_\epsilon$ and in $\Theta_\epsilon$ for fixed $B$; hence, it is bi-convex, which in turn
implies that the proposed algorithm may fail to converge to the global optimum, especially in settings where $p_1>n$, as pointed out by \citet{lee2012simultaneous}. As is the case with most non-convex problems, good initial parameters are beneficial for fast 
convergence of the algorithm, a fact supported by our numerical work on the present problem.  
Further, a good initialization is critical in establishing convergence of the algorithm for this problem (see Section~\ref{sec:convergence}).
To that end, we introduce a {\em screening step} for obtaining a good initial estimate for $B$. The theoretical 
justification for employing the screening step is provided in Section~\ref{sec:FWER}.

An outline of the computational procedure is presented in Algorithm~\ref{alg:as}, while the details of each step involved are discussed next.

\begin{algorithm}[t]
\SetKwData{Left}{left}\SetKwData{This}{this}\SetKwData{Up}{up}
\SetKwFunction{Union}{Union}\SetKwFunction{FindCompress}{FindCompress}
\SetKwInOut{Input}{Input}\SetKwInOut{Output}{Output}
\caption{Computational procedure for estimating $B$ and $\Theta_\epsilon$}\label{alg:as}

\Input{Data from the parent layer $X$ and the response layer $Y$.}
\BlankLine
\textbf{Screening:}\\
\hspace*{5mm}\begin{minipage}[t]{14.5cm}
 \For{$j=1,\cdots,p_2$}{regress $Y_j$ on $X$ using the de-biased Lasso procedure in \citet{javanmard2014confidence} and obtain the corresponding vector of $p$-values $P_j$\;}
 obtain adjusted $p$-values $\widetilde{P}_j$ by applying Bonferroni correction to  $\mathrm{vec}(P_1,\cdots,P_j)$\;
 determine the support set $\mathcal{B}_j$ for each regression using (\ref{eqn:support}).
\end{minipage}
\BlankLine
\textbf{Initialization:}\\
\hspace*{5mm}\begin{minipage}[t]{14.5cm}
Initialize column $j=1,\cdots,p_2$ of $\widehat{B}^{(0)}$ by solving (\ref{eqn:initialB}).\\
Initialize $\widehat{\Theta}_\epsilon^{(0)}$ by solving (\ref{eqn:initialTheta}) using the graphical lasso \citep{friedman2008sparse}.
\end{minipage}\\
\While{$|f(\widehat{B}^{(k)},\widehat{\Theta}_\epsilon^{(k)})-f(\widehat{B}^{(k+1)},\widehat{\Theta}_\epsilon^{(k+1)})|\geq \epsilon$}{
update $\widehat{B}$ with (\ref{eqn:updateB})\;
update $\widehat{\Theta}_\epsilon$ with (\ref{eqn:updateTheta})\;
}
\BlankLine
\textbf{Refitting $B$ and $\Theta_\epsilon$:}
\hspace*{5mm}\begin{minipage}[t]{14.5cm}
 \For{$j=1,\cdots,p_2$}{Obtain the refitted $\widetilde{B}_j$ using (\ref{eqn:refitB})\;}
 re-estimate $\widetilde{\Theta}_\epsilon$ using (\ref{glasso-final}) with $W$ coming from stability selection.
\end{minipage} 
\BlankLine
\Output{Final Estimates $\widetilde{B}$ and $\widetilde{\Theta}_\epsilon$.}
\end{algorithm}

\textit{\textbf{Screening.}} For each variable $Y_j, j=1,\cdots,p_2$ in the response layer, regress $Y_j$ on $X$ via the de-biased Lasso procedure proposed by \citet{javanmard2014confidence}. The output consists of the $p$-value(s) for each predictor in each regression, denoted by $P_j$, with $P_j\in
[0,1]^{p_1}$. To control the family-wise error rate of the estimates, we do a Bonferroni correction at level $\alpha$: define $\alpha^{\star} = \alpha/p_1 p_2$ and set $B_{j,k} = 0$ if the
$p$-value obtained for the $k$'th predictor in the $j$'th regression $P_{j,k}$ exceeds $\alpha^{\star}$. Further, let 
\begin{equation}\label{eqn:support}
\mathcal{B}_j=\{B_j\in\mathbb{R}^{p_1}:B_{j,k}=0\text{ if }k\in \widehat{S}_j^c\} \subseteq \mathbb{R}^{p_1},
\end{equation}
where $\widehat{S}_j$ is the collection of indices for those predictors deemed ``active" for response $Y_j$:
\begin{equation*}
\widehat{S}_j = \{k:P_{j,k}>\alpha^{\star}\}, \quad \text{for }j=1,\cdots,p_2.
\end{equation*}
Therefore, subsequent estimation of the elements of $B$ will be restricted to $\mathcal{B}_1\times\cdots\times \mathcal{B}_{p_2}$. \\

\textit{\textbf{Alternating Search.}} In this step, we utilize the bi-convexity of the problem and estimate $B$ and $\Theta_\epsilon$ by minimizing in an iterative fashion the objective function with respect to (w.r.t.) one set of parameters, while holding the other set fixed within each iteration. 

As with most iterative algorithms, we need an initializer; for $\widehat{B}^{(0)}$ it corresponds to a Lasso/Ridge regression estimate with a small penalty, while for $\widehat{\Theta}_\epsilon$ we use the Graphical Lasso procedure applied to the residuals obtained from the first stage regression. That is, for each $j=1,\cdots,p_2$, 
\begin{equation}\label{eqn:initialB}
\widehat{B}_j^{(0)} = \argmin\limits_{B_j\in\mathcal{B}_j} \left\{ \|Y_j - XB_j\|^2_2 + \lambda_n^0 \|B_j\|_1\right\},
\end{equation}
where $\lambda_n^0$ is some small tuning parameter for initialization, and set  $\widehat{E}^{(0)}_j := Y_j-X\widehat{B}_j^{(0)}$. An initial estimate for $\widehat{\Theta}_\epsilon$ is then given by solving for the following optimization problem with the graphical lasso \citep{friedman2008sparse} procedure:
 \begin{equation*}\label{eqn:initialTheta}
 \widehat{\Theta}_\epsilon^{(0)} = \argmin\limits_{\Theta_\epsilon\in\mathbb{S}_{++}^{p_2\times p_2}} \left\{
 \log\det\Theta_\epsilon - \text{tr}(\widehat{S}^{(0)}\Theta_\epsilon) + \rho_n\|\Theta_\epsilon\|_{1,\text{off}}
 \right\},
 \end{equation*}
where $\widehat{S}^{(0)}$ is the sample covariance matrix based on $(\widehat{E}^{(0)}_1,\cdots,\widehat{E}^{(0)}_{p_2})$.\\

Next we use an alternating block coordinate descent algorithm with $\ell_1$ penalization to reach a stationary point of the objective function (\ref{eqn:obj2}):
\begin{itemize}
\setlength{\itemsep}{0.2pt}
\item[--] Update $B$ as:
\begin{equation}\label{eqn:updateB}
\widehat{B}^{(k+1)} = \argmin\limits_{B\in \mathcal{B}_1\times \cdots\times \mathcal{B}_{p_2}} \left\{ \frac{1}{n}\sum_{i=1}^{p_2}\sum_{j=1}^{p_2}(\widehat{\sigma}_{\epsilon}^{ij})^{(k)}(Y_i-XB_i)^\top (Y_j - XB_j) + \lambda_n\sum_{j=1}^{p_2}\|B_j\|_1 \right\},
\end{equation}
which can be obtained by cyclic coordinate descent w.r.t each column $B_j$ of $B$, that is, update each column $B_j$ by: 
\begin{equation}\label{eqn:Bjunrestricted}
\widehat{B}_j^{(t+1)} = \argmin\limits_{B_j\in\mathcal{B}_j} \begin{Bmatrix} \frac{(\widehat{\sigma}^{jj}_\epsilon)^{(k)}}{n}\|Y_j+r_j^{{(t+1)}}-XB_j\|_2^2 + \lambda_n\|B_j\|_1 \end{Bmatrix},
\end{equation}
 where 
$$r_j^{(t+1)} = \frac{1}{(\widehat{\sigma}^{jj}_\epsilon)^{(k)}}\left[ \sum_{i=1}^{j-1} (\widehat{\sigma}^{ij}_\epsilon)^{(k)}(Y_i-X\widehat{B}_i^{(t+1)}) + \sum_{i=j+1}^{p_2} (\widehat{\sigma}^{ij}_\epsilon)^{(k)} (Y_i-X\widehat{B}_i^{(t)})\right],$$
and iterate over all columns until convergence. Here, we use $k$ to index the outer iteration while minimizing w.r.t. $B$ or $\Theta_\epsilon$, and use $t$ to index the inner iteration while cyclically minimizing w.r.t. each column of $B$. 
\item[--] Update  $\Theta_\epsilon$ as:
\begin{equation}\label{eqn:updateTheta}
 \widehat{\Theta}_\epsilon^{(k+1)} = \argmin\limits_{\Theta_\epsilon\in\mathbb{S}_{++}^{p_2\times p_2}} \left\{
 \log\det\Theta_\epsilon - \text{tr}(\widehat{S}^{(k+1)}\Theta_\epsilon) + \rho_n\|\Theta_\epsilon\|_{1,\text{off}}
 \right\},
 \end{equation}
 where $\widehat{S}^{(k+1)}$ is the sample covariance matrix based on $\widehat{E}^{(k+1)}_j = Y_j - X\widehat{B}_j^{(k+1)},j=1,\cdots,p_2$.
\end{itemize}

\textit{\textbf{Refitting and Stabilizing.}} As noted in the introduction, this step is beneficial in applications, especially when one deals with large scale multi-layer networks and relatively smaller sample sizes. Denote the solution obtained by the above iterative procedure by $B^{\infty}$ and $\Theta_\epsilon^\infty$. For each $j=1,\cdots,p_2$, set $\widetilde{\mathcal{B}}_j = \{B_j: B_{j,i}=0 \text{ if }B^\infty_{j,i}=0,B_j\in\mathbb{R}^{p_1}\}$ and the final estimate for $B_j$ is given by ordinary least squares: 
\begin{equation}\label{eqn:refitB}
\widetilde{B}_j = \argmin\limits_{B_j \in \widetilde{\mathcal{B}}_j} \|Y_j - XB_j\|^2.
\end{equation}
For $\Theta_\epsilon$, we obtain the final estimate by a combination of stability selection \citep{meinshausen2010stability} and graphical lasso \citep{friedman2008sparse}. That is, after obtaining the refitted residuals $\widetilde{E}_j:=Y_j-X\widetilde{B}_j,j=1,\cdots,p_2$, based on the stability selection procedure with the graphical lasso, we obtain the stability path, or probability matrix $W$ for each edge, which records the proportion of each edge being selected based on bootstrapped samples of $\widetilde{E}_j$'s. Then, using this probability matrix $W$ as a weight matrix, we obtain the final estimate of $\widetilde{\Theta}_\epsilon$ as follow:
\begin{equation}\label{glasso-final}
\widetilde{\Theta}_\epsilon = \argmin\limits_{\Theta_\epsilon\in\mathbb{S}^{p_2\times p_2}_{++}}\left\{ \log\det \Theta_\epsilon - \mbox{tr}(\widetilde{S}\Theta_\epsilon) + \widetilde{\rho}_n\|(1-W)*\Theta_\epsilon\|_{1,\text{off}} \right\},
\end{equation} 
where we use $*$ to denote the element-wise product of two matrices, and $\widetilde{S}$ is the sample covariance matrix based on the refitted residuals $\widetilde{E}$. Again, (\ref{glasso-final}) can be solved by the graphical lasso procedure \citep{friedman2008sparse}, with $\widetilde{\rho}_n$ properly chosen. 

%
%

\subsection{\normalsize Tuning Parameter Selection.}

To select the tuning parameters $(\lambda_n,\rho_n)$, we use the Bayesian Information Criterion(BIC), which is the summation of a goodness-of-fit term (log-likelihood) and a penalty term. The explicit form of BIC (as a function of $B$ and $\Theta_\epsilon$) in our setting is 
given by 
\begin{equation*}
\text{BIC}(B,\Theta_\epsilon) = -\log\det\Theta_\epsilon + \text{tr}(S\Theta_\epsilon) + \frac{\log n}{n} (\frac{\|\Theta_\epsilon\|_0-p_2}{2} + \|B\|_0)
\end{equation*}
where 
\begin{equation*}
S :=  \frac{1}{n}\begin{bmatrix}
(Y_1-XB_1)^\top \\  \vdots \\ (Y_{p2}-XB_{p_2})^\top  
\end{bmatrix} \begin{bmatrix}
(Y_1-XB_1) & \cdots & (Y_{p_2}-XB_{p_2})
\end{bmatrix},
\end{equation*}
and $\|\Theta_\epsilon\|_0$ is the total number of nonzero entries in $\Theta_\epsilon$. Here we penalize the non-zero elements in the upper-triangular part of $\Theta_\epsilon$ and the non-zero ones in 
$B$. We choose the combination $(\lambda_n^*,\rho_n^*)$ over a grid of $(\lambda,\rho)$ values, and $(\lambda_n^*,\rho_n^*)$ should minimize the BIC evaluated at $(B^{\infty},\Theta_\epsilon^{\infty})$.  \\

%
%

\section{Theoretical Results} \label{sec:Theory}

In this section, we establish a number of theoretical results for the proposed iterative algorithm. We focus the presentation on the two-layer structure,
since as explained in the previous section the multi-layer estimation problem decomposes to a series of two-layer ones. 
As mentioned in the introduction, one key challenge for estabilishing the theoretical results comes from the fact that the objective function (\ref{eqn:obj2}) is not jointly convex in $B$ and $\Theta_\epsilon$. Consequently, if we simply used properties of block-coordinate descent algorithms, we would not be able to provide the necessary theoretical guarantees for the estimates we obtain. On the other hand, the biconvex nature of the objective function allows
us to establish convergence of the alternating algorithm to a stationary point, provided it is initialized from a point close enough to the true parameters.
This can be accomplished using a Lasso-based initializer for $B$ and $\Theta_\epsilon$ as previously discussed. The details of algorithmic convergence 
are presented in Section~\ref{sec:convergence}.

Another technical challenge is that each update in the alternating search step relies on estimated quantities --namely the regression and precision matrix parameters --rather than the raw data, whose estimation precision needs to be controlled {\em uniformly} across all iterations. The details of establishing
consistency of the estimates for both fixed and random realizations are given in Section~\ref{sec:consistency}. 

Next, we outline the structure of this section. In Section \ref{sec:convergence} Theorem~\ref{thm:convergence}, we show that for any fixed set of realization of $X$ and $E$\footnote{We actually observe $X$ and $Y$, which is given by a corresponding set of realization in $X$ and $E$ based on the model.}, the iterative algorithm is guaranteed to converge to a stationary point if estimates for all iterations lie in a compact ball around the true value of the parameters. In Section \ref{sec:consistency}, we show in Theorem~\ref{thm:beta-theta-bound} that for any random $X$ and $E$, with high probability, the estimates for all iterations lie in a compact ball around the true value of the parameters. Then in Section~\ref{sec:FWER}, we show that asymptotically with $\log(p_1p_2)/n\rightarrow 0$, while keeping the family-wise type I error under some pre-specified level, the screening step correctly identifies the true support set for each of the regressions, based upon which the iterative algorithm is provided with an initializer that is close to the true value of the parameters. Finally in Section \ref{sec:identifiability}, we provide sufficient conditions for both directed and undirected edges to be 
identifiable (estimable) for multi-layered network.

Throughout this section, to distinguish the estimates from the true values, we use $B^*$ and $\Theta_\epsilon^*$ to denote the true values.

%
%

\subsection{\normalsize Convergence of the Iterative Algorithm}\label{sec:convergence}

In this subsection, we prove that the proposed block relaxation algorithm converges to a stationary point for any fixed set of data, provided that the estimates for all iterations lie in a compact ball around the true value of the parameters. This requirement is shown to be satisfied with high probability in the next subsection \ref{sec:consistency}.

Decompose the optimization problem in (\ref{eqn:obj2}) as follows:
\begin{equation*}
\min\limits_{\substack{B\in\mathbb{R}^{p_1\times p_2}\\ \Theta_\epsilon\in\mathbb{S}_{++}^{p_2\times p_2}}}f(B,\Theta_\epsilon) \equiv f_0(B,\Theta_\epsilon) + f_1(B) + f_2(\Theta_\epsilon) 
\end{equation*}
where 
\begin{equation*}
f_0(B,\Theta_\epsilon) =\frac{1}{n}\sum_{j=1}^{p_2}\sum_{i=1}^{p_2}\sigma^{ij}_\epsilon (Y_i-XB_i)'(Y_j - XB_j)-\log \det \Theta_\epsilon = \text{tr}(S\Theta_\epsilon) - \log \det \Theta_\epsilon, 
\end{equation*}
\begin{equation*}
f_1(B) = \lambda_n\|B\|_1, \quad f_2(\Theta_\epsilon) = \rho_n\|\Theta_\epsilon\|_{1,\text{off}}.
\end{equation*}
and $\mathbb{S}_{++}^{p_2\times p_2}$ is the collection of $p_2\times p_2$ symmetric positive definite matrices. Further, denote the limit point (if there is any) of $\{\widehat{B}^{(k)}\}$ and $\{\widehat{\Theta}_\epsilon^{(k)}\}$ by 
$B^\infty = \lim_{k\rightarrow\infty} \widehat{B}^{(k)}$ and $\Theta^\infty_\epsilon =  \lim_{k\rightarrow\infty} \widehat{\Theta}_\epsilon^{(k)}$, respectively. 

\begin{definition}[stationary point\citep{tseng2001convergence} pp.479] Define $z$ to be a stationary point of $f$ if $z\in \dom (f)$ and $f'(z;d)\geq 0,\forall \text{ direction }d=(d_1,\cdots,d_N)$ where $d_t$ is the $t^{\text{th}}$ coordinate block. 
\end{definition}

\begin{definition}[Regularity \citep{tseng2001convergence} pp.479] $f$ is regular at $z\in\dom(f)$ if $f'(z;d)\geq 0$ for all $d=(d_1,\cdots,d_N)$ such that 
\begin{equation*}
f'(z;(0,\cdots,d_t,\cdots,0))\geq 0,\qquad t=1,2,\cdots,N.
\end{equation*}  
\end{definition}

\begin{definition}[Coordinate-wise minimum] Define $(B^\infty,\Theta_\epsilon^\infty)$ to be a coordinate-wise minimum if 
\begin{eqnarray*}
f(B^{\infty},\Theta_\epsilon) &\geq& f(B^\infty,\Theta_\epsilon^\infty), \quad \forall \Theta_\epsilon\in \mathbb{S}_{++}^{p_2\times p_2}, \\
f(B,\Theta_\epsilon^{\infty})&\geq & f(B^{\infty},\Theta_\epsilon^{\infty}),\quad \forall B\in\mathbb{R}^{p_1\times p_2}.
\end{eqnarray*}
\end{definition}
Note for our iterative algorithm, we only have two blocks, hence with the above notation, $N=2$. 

\begin{remark}
\citet{tseng2001convergence} proved that if $f$ satisfies certain conditions \citep[see Theorem 4.1 (a), (b) and (c) for details]{tseng2001convergence}, the limit point given by the general block-coordinate descent algorithm (with $N\geq 2$ blocks) is a stationary point of $f$. However, in the high dimensional setting, the posited objective function does not satisfy {\em any} of the assumptions in that Theorem. Hence, for this problem, we need to employ a different strategy to prove convergence to a stationary point, and the resulting statements hold true for all problems that use a $2$-block coordinate descent algorithm.
\end{remark}
\medskip
Since $\dom(f_0)$ is open and $f_0$ is G\^{a}teaux-differentiable on the $\dom(f_0)$, by \citet{tseng2001convergence} Lemma 3.1, $f$ is regular in the $\dom (f)$. From the discussion on Page 479 of \citep{tseng2001convergence}, we then have: 
\newline
\newline
{\bf Fact 1:} 
Every coordinate-wise minimum is a stationary point of $f$.
\newline
\newline
The following theorem shows that any limit point $(B^\infty,\Theta_\epsilon^\infty)$ of the iterative algorithm described in Section~\ref{sec:estimation} is a stationary point of $f$, as long as all the iterates are within a closed ball around the truth.

\begin{theorem}[Convergence for fixed design]\label{thm:convergence}
Suppose for any fixed realization of $X$ and $E$, the estimates $\left\{(\widehat{B}^{(k)},\widehat{\Theta}_\epsilon^{(k)})\right\}_{k=1}^\infty$ obtained by implemeting the alternating search step satisfy the following bound:
\begin{equation*}
\fnorm{(\widehat{B}^{(k)},\widehat{\Theta}_\epsilon^{(k)})-(B^*,\Theta^*)}\leq R,\quad \text{for some }R>0,~\forall k\geq 1.
\end{equation*}
Then any limit point $(B^\infty,\Theta_\epsilon^\infty)$ of the iterative algorithm is a stationary point of $f$. 
\end{theorem}
\begin{proof}
We initialize the algorithm at $(\widehat{B}^{(0)},\widehat{\Theta}^{(0)}_\epsilon)\in\dom(f)$. Then for all $k\geq 1$:
\begin{eqnarray}
\widehat{B}^{(k)} &=& \argmin\limits_{B} f(B,\widehat{\Theta}_\epsilon^{(k-1)}) \label{getB}\\
\widehat{\Theta}^{(k)}_\epsilon & = & \argmin\limits_{\Theta_\epsilon} f(\widehat{B}^k,\Theta_\epsilon)\label{getTheta}
\end{eqnarray} 
Now, consider a limit point $(B^\infty,\Theta_\epsilon^\infty)$ of the sequence $\{(\widehat{B}^{(k)},\widehat{\Theta}_\epsilon^{(k)})\}_{k\geq 1}$. Note that such limit point exists by Bolzano-Weierstrass theorem since the sequence $\{(\widehat{B}^{(k)},\widehat{\Theta}_\epsilon^{(k)})\}_{k\geq 1}$ is bounded. Consider a subsequence $\mathcal{K}\subseteq\{1,2,\cdots\}$ such that $(\widehat{B}^{(k)},\widehat{\Theta}_\epsilon^{(k)})_{k\in \mathcal{K}}$ converges to $(B^\infty,\Theta_\epsilon^{\infty})$. Now for the bounded sequence $\{(\widehat{B}^{(k+1)},\widehat{\Theta}_\epsilon^{(k)})\}_{k\in\mathcal{K}}$, without loss of generality\footnote{switching to some further subsequence of $\mathcal{K}$ if necessary.}, we can say that 
\begin{equation*}
\{(\widehat{B}^{(k+1)},\widehat{\Theta}_\epsilon^{(k)})\}_{k\in\mathcal{K}}\rightarrow (\widetilde{B}^\infty,\widetilde{\Theta}_\epsilon^\infty), \quad \text{for some }(\widetilde{B}^\infty,\widetilde{\Theta}_\epsilon^\infty)\in\dom(f).
\end{equation*}
By (\ref{getB}) it follows immediately that $\widetilde{\Theta}_\epsilon^\infty = \Theta^\infty_\epsilon$. Also, the following inequality holds: 
\begin{equation*}
f(\widehat{B}^{(k+1)},\widehat{\Theta}_\epsilon^{(k+1)})\leq f(\widehat{B}^{(k+1)},\widehat{\Theta}_\epsilon^{(k)})\leq f(\widehat{B}^{(k)},\widehat{\Theta}_\epsilon^{(k)}).
\end{equation*}
Thus, by letting $k\rightarrow\infty$ over $\mathcal{K}$, we have 
\begin{equation*}
f(B^\infty,\Theta_\epsilon^\infty)\leq f(\widetilde{B}^\infty,\Theta_\epsilon^\infty)\leq f(B^\infty,\Theta_\epsilon^\infty),
\end{equation*}
since $f$ is continuous. This implies that 
\begin{equation} \label{equal}
f(\widetilde{B}^\infty,\Theta_\epsilon^\infty) = f(B^\infty,\Theta_\epsilon^\infty)
\end{equation}
Next, since $f(\widehat{B}^{(k+1)},\widehat{\Theta}_\epsilon^{(k)})\leq f(B,\widehat{\Theta}_\epsilon^{(k)})$, for all $B\in \mathbb{R}^{p_1\times p_2}$, let $k$ grow along $\mathcal{K}$, and we obtain the following:
\begin{equation*}
f(\widetilde{B}^\infty,\Theta_\epsilon^\infty)\leq f(B,\Theta_\epsilon^\infty),\quad \forall B\in \mathbb{R}^{p_1\times p_2}.
\end{equation*}
It then follows from (\ref{equal}) that 
\begin{equation}\label{unequal-1}
f(B^\infty,\Theta_\epsilon^\infty)\leq f(B,\Theta_\epsilon^\infty), \quad \forall B\in\mathbb{R}^{p_1\times p_2}.
\end{equation}
Finally, note that $f(\widehat{B}^{(k)},\widehat{\Theta}_\epsilon^{(k)})\leq f(\widehat{B}^{(k)},\Theta_\epsilon)$, for all $\Theta\in\mathbb{S}^{p_2\times p_2}_{++}$. As before, let $k$ grow along $\mathcal{K}$ and with the continuity of $f$, we obtain:
\begin{equation}\label{unequal-2}
f(B^\infty,\Theta_\epsilon^\infty) \leq f(B^\infty,\Theta_\epsilon), \quad \forall \Theta_\epsilon\in\mathbb{S}_{++}^{p_2\times p_2}. 
\end{equation}
Now, (\ref{unequal-1}) and (\ref{unequal-2}) together imply that $(B^\infty,\Theta_\epsilon^\infty)$ is a coordinate-wise minimum of $f$ and by Fact 1, 
also a stationary point of $f$. 
\end{proof}

\begin{remark}
Recall that in classical parametric statistics, MLE-type asymptotics are derived
after establishing that with probability tending to $1$ as the sample
size $n$ goes to infinity, the likelihood equation has a sequence of
roots (hence stationary points of the likelihood function)  that converges in probability to the true value. Any such sequence of roots is shown to be asymptotically normal and efficient. Note that such (a sequence of) roots may not be global maximizers since parametric likelihoods are not globally log-concave \cite[see Chapter 6][]{lehmann1998theory}. Here we show that the $(B^\infty,\Theta_\epsilon^\infty)$ obtained by the iterative algorithm is a
stationary point which satisfies the first-order condition for being a maximizer of the penalized log--likelihood function (which is just the negative of the penalized 
least--squares function). Moreover, if we let $n$ go to infinity, $(B^\infty,\Theta_\epsilon^\infty)$ converges to the true value in probability (shown in Theorem~\ref{thm:beta-theta-bound}), and therefore behaves the same as the sequence of roots in the classical parametric problem alluded to above. Thus, while 
$(B^\infty,\Theta_\epsilon^\infty)$ may not be the global maximizer, it can, nevertheless, to all intents and purposes, be deemed as the MLE.
\end{remark}

\begin{remark}
The above convergence result is based upon solving the optimization problem on the ``entire" space, that is, we don't restrict $B$ to live in any subspace. However, when actually implementing the proposed computational procedure, the optimization of the $B$ coordinate is restricted to $\mathcal{B}_1\times\cdots\times \mathcal{B}_{p_2}$ (as defined in eqn.\ref{eqn:support}). It should be noted that the same convergence property still holds, since for all $k\geq 1$, the following bound holds, for some $R'>0$: 
\begin{equation}\label{eqn:bound}
\left\|\left(\widehat{B}^{(k)}_{\text{restricted}},\widehat{\Theta}^{(k)}_\epsilon\right) - \left(B^*,\Theta_\epsilon^*\right)\right\|_\text{F}\leq R'.
\end{equation}
Consequently, the rest of the derivation in Theorem~\ref{thm:convergence} follows, leading to the convergence property. The bound in eqn (\ref{eqn:bound}) will be shown at the end of Section~\ref{sec:consistency}.  
\end{remark} 

\subsection{\normalsize Estimation consistency}\label{sec:consistency}

In this subsection, we show that given a random realization of $X$ and $E$, with high probability, the sequence $\left\{(\widehat{B}^{(k)},\widehat{\Theta}_\epsilon^{(k)})\right\}_{k=1}^\infty$ lies in a non-expanding ball around $(B^*,\Theta_\epsilon^*)$, thus satisfying the
condition of Theorem ~\ref{thm:convergence} for convergence of the alternating algorithm.

It should be noted that for the alternating search procedure, we restrict our estimation on a subspace identified by the screening step. However, for the remaining of this subsection, the main propositions and theorems are based on the procedure without such restriction, i.e., we consider ``generic" regressions on the entire space of dimension $p_1\times p_2$. Notwithstanding, it can be easily shown that the theoretical results for the regression
parameters on a restricted domain follow easily from the generic case, as explained in Remark~\ref{remark:restriction}.

Before providing the details of the main theorem statements and proofs, we first introduce additional notations. Let $\beta = \mathrm{vec}(B)$ be the vectorized version of the regression coefficient matrix. Correspondingly, we have $\widehat{\beta}^{(k)}=\mathrm{vec}(\widehat{B}^{(k)})$ and $\beta^* = \mathrm{vec}(B^*)$. Moreover, we drop the superscripts and use $\widehat{\beta}$ and $\widehat{\Theta}_\epsilon$ to denote the generic estimators given by (\ref{est-1}) and (\ref{est-2}), as opposed to those obtained in any specific iteration:
\begin{eqnarray}
\widehat{\beta} & \equiv & \argmin\limits_{\beta\in\mathbb{R}^{p_1p_2}}\left\{ -2\beta'\widehat{\gamma} + \beta'\widehat{\Gamma}\beta + \lambda_n\|\beta\|_1\right\}, \label{est-1}\\
\widehat{\Theta}_\epsilon & \equiv & \argmin\limits_{\Theta_\epsilon\in\mathbb{S}^{p_2\times p_2}_{++}} \left\{ -\log\det\Theta_\epsilon + \text{tr}\left( \widehat{S}\Theta_\epsilon\right) + \rho_n \|\Theta_\epsilon\|_{1,\text{off}}\right\}, \label{est-2}
\end{eqnarray}
where 
\begin{equation*}
\widehat{\Gamma} = \left(\widehat{\Theta}_\epsilon \otimes \frac{X'X}{n}\right), \ \widehat{\gamma} = \left(\widehat{\Theta}_\epsilon \otimes X' \right) \mathrm{vec}(Y)/n, \ \widehat{S} = \frac{1}{n}\left( Y - X\widehat{B}\right)'\left( Y - X\widehat{B}\right).
\end{equation*}

\begin{remark} As opposed to (\ref{est-1}) and (\ref{est-2}), if $\widehat{\gamma}$ and $\widehat{\Gamma}$ are replaced by plugging in the true values of the parameters, the two problems in (\ref{est-1}) and (\ref{est-2}) become:
\begin{eqnarray}
\bar{\beta} & \equiv & \argmin\limits_{\beta\in\mathbb{R}^{p_1p_2}}\left\{ -2\beta'\bar{\gamma} + \beta'\bar{\Gamma}\beta + \lambda_n\|\beta\|_1\right\}, \label{est-1.1}\\
\bar{\Theta}_\epsilon & \equiv & \argmin\limits_{\Theta_\epsilon\in \mathbb{S}^{p_2\times p_2}_{++}} \left\{ -\log\det\Theta_\epsilon + \text{tr}\left( S\Theta_\epsilon\right) + \rho_n \|\Theta_\epsilon\|_{1,\text{off}}\right\}, \label{est-2.1}
\end{eqnarray}     
where 
\begin{equation*}
\bar{\Gamma} = \left(\Theta^*_\epsilon \otimes \frac{X'X}{n}\right), ~~ \bar{\gamma} = \left( \Theta^*_\epsilon \otimes X'\right) \mathrm{vec}(Y)/n, ~~ S = \frac{1}{n}\left( Y - XB^*\right)'\left( Y - XB^*\right)\equiv\widehat{\Sigma}_\epsilon.
\end{equation*}
In (\ref{est-1.1}), we obtain $\beta$ using a penalized maximum likelihood regression estimate, and (\ref{est-2.1}) corresponds to the generic setting for using the graphical Lasso. A key difference between the estimation problems in (\ref{est-1}) and (\ref{est-2}) versus those in (\ref{est-1.1}) and (\ref{est-2.1}) is that to obtain $\widehat{\beta}$ and $\widehat{\Theta}_\epsilon$ we use {\em estimated quantities} rather than the raw data. This is exactly how we implement our iterative algorithm, namely, we obtain $\widehat{\beta}^{(k)}$ using $\widehat{S}^{(k-1)}$ as a surrogate for the sample covariance of the true error (which is unavailable), then estimate $\widehat{\Theta}_\epsilon^{(k)}$ using the information in $\widehat{\beta}^{(k)}$. This adds complication for establishing the consistency results. Original consistency results for the estimation problem in (\ref{est-1.1}) and (\ref{est-2.1}) are available in \citet{basu2015estimation} and \citet{ravikumar2011high}, respectively. Here we borrow ideas from corresponding theorems in those two papers, but need to tackle concentration bounds of relevant quantities with additional care. This part of the result and its proof are shown in Theorem~\ref{thm:beta-theta-bound}. 
\end{remark}

As a road map toward our desired result established in Theorem~\ref{thm:beta-theta-bound}, we first show in Theorem~\ref{thm:ErrorBound_beta} that for any fixed realization of $X$ and $E$, under a number of conditions on (or related to) $X$ and $E$,  when $\|\widehat{\Theta}_\epsilon-\Theta_\epsilon^*\|_\infty$ is small (up to a certain order), the error of $\widehat{\beta}$ is well-bounded. We then verify in Proposition \ref{prop:REcondition} and \ref{prop:deviation} that for random $X$ and $E$, the above-mentioned conditions hold with high probability. Similarly in Theorem~\ref{thm:ErrorBound_Theta}, we show that for fixed realizations in $X$ and $E$, under certain conditions (verified for random $X$ and $E$ in Proposition~\ref{prop:residual-concentration}), the error of $\widehat{\Theta}_\epsilon$ is also well-bounded, given $\|\widehat{\beta}-\beta^*\|_1$ being small. Finally in Theorem~\ref{thm:beta-theta-bound}, we show that for random $X$ and $E$, with high probability, the iterative algorithm gives $\{(\widehat{\beta}^{(k)},\Theta^{(k)}_\epsilon)\}$ that lies in a small ball centered at $(\beta^*,\Theta_\epsilon^*)$, whose radius depends on $p_1$, $p_2$, $n$ and the sparsity levels. 

Next, for establishing the main propositions and theorems, we introduce some additional notations:
\begin{itemize}
\setlength\itemsep{0pt}
\item[--] Sparsity level of $\beta^*$: $s^{**} := \|\beta^*\|_0 = \sum_{j=1}^{p_2}\|B_j^*\|_0 = \sum_{j=1}^{p_2}s_j^*$. As a reminder of the previous notation, we have $s^*=\max\limits_{j=1,\cdots,p_2}s_j^*$.
\item[--] True edge set of $\Theta_\epsilon^*$: $S^*_\epsilon$, and let $s^*_\epsilon:=|S^*_\epsilon|$ be its cardinality.
\item[--] Hessian of the log-determinant barrier $\log\det\Theta$ evaluated at $\Theta_\epsilon^*$: \[H^*:=\frac{\mbox{d}^2}{\mbox{d}\Theta^2}\log\Theta\big{|}_{\Theta^*_\epsilon} = \Theta^{*-1}_\epsilon\otimes \Theta^{*-1}_\epsilon.\]
\item[--] Matrix infinity norm of the true error covariance matrix $\Sigma_\epsilon^*$: \[\kappa_{\Sigma^*_\epsilon}:=\vertiii{\Sigma^*_\epsilon}_\infty =  \max\limits_{i=1,2,\cdots,p_2}\sum_{j=1}^{p_2} |\Sigma_{\epsilon,ij}^*|.\]
\item[--] Matrix infinity norm of the Hessian restricted to the true edge set: 
\[\kappa_{H^*} := \vertiii{(H^*_{S^*_\epsilon S^*_\epsilon})}_\infty =\max\limits_{i=1,2,\cdots,p_2}\sum_{j=1}^{p_2} \left|H^*_{S^*_\epsilon S^*_\epsilon,ij}\right|.\] 
\item[--] Maximum degree of $\Theta^*_\epsilon$: $d:= \max\limits_{i=1,2,\cdots,p_2}\|\Theta_{\epsilon,i\cdot}^*\|_0$.
\item[--] We write $A\gtrsim B$ if there exists some absolute constant $c$ that is independent of the model parameters such that $A\geq cB$. 
\end{itemize}

\begin{definition}[Incoherence condition \citep{ravikumar2011high}]\label{def:incoherence} $\Theta^*_\epsilon$ satisfies the incoherence condition if:
\begin{equation*}
\max\limits_{e\in (S_\epsilon^*)^c}\| H^*_{eS^*_\epsilon}(H^*_{S_\epsilon^* S_\epsilon^*})^{-1}\|_1 \leq 1-\xi, \quad \text{for some }\xi\in(0,1).
\end{equation*}
\end{definition}
\begin{definition}[Restricted eigenvalue (RE) condition \citep{loh2011high}] A symmetric matrix $A\in\mathbb{R}^{m\times m}$ satisfies the RE condition with curvature $\varphi>0$ and tolerance $\phi>0$, denoted by $A\sim RE(\varphi,\phi)$ if 
\begin{equation*}
\theta'A\theta\geq \varphi\|\theta\|^2 - \phi\|\theta\|_1^2, \quad \forall \theta\in\mathbb{R}^m.
\end{equation*}
\end{definition}

\begin{definition}[Diagonal dominance] A matrix $A\in\mathbb{R}^{m\times m}$ is strictly diagonally dominant if 
\begin{equation*}
|a_{ii}|>\sum_{j\neq i}|a_{ij}|,\quad \forall i=1,\cdots,m.
\end{equation*}
\end{definition}

Based on the model in Section~\ref{sec:set-up}, since we are assuming $\bs{X}=(X_1,\cdots,X_{p_1})'$ and $\bs{\epsilon}=(\epsilon_1,\cdots,\epsilon_{p_2})$ come from zero-mean Gaussian distributions, it follows that $\bs{X}$ and $\bs{\epsilon}$ are zero-mean sub-Gaussian random vectors with parameters $(\Sigma_X,\sigma_x^2)$ and $(\Sigma_\epsilon^*,\sigma_\epsilon^2)$, respectively. Moreover, thoughout this section, all results are based on the assumption that $\Theta_\epsilon^*$ is diagnally dominant.

\begin{remark}
Before moving on to the main statements of Theorem~\ref{thm:ErrorBound_beta}, we would like to point out that with a slight abuse of notation, for Theorem~\ref{thm:ErrorBound_beta} and its related propositions and corollaries, the statements and analyses are based on equation~(\ref{est-1}) only, with \textit{any determinisitic symmetric matrix} $\widehat{\Theta}_\epsilon$ within a small ball around $\Theta^*_\epsilon$. Similarly in Theorem~\ref{thm:ErrorBound_Theta}, Proposition~\ref{prop:residual-concentration} and Corollary~\ref{cor:erorrboundTheta}, the analyses are based on equation~（(\ref{est-2}) only, for {\em any given determinisitic}  $\widehat{\beta}$ within a small ball around $\beta^*$. The randomness of $\widehat{\beta}$ and $\widehat{\Theta}_\epsilon$ during the iterative procedure will be taken into consideration comprehensively in Theorem~\ref{thm:beta-theta-bound}. 
\end{remark}


%
%
%
%

\begin{theorem}[Error bound for $\widehat{\beta}$ with fixed realizations of $X$ and $E$] \label{thm:ErrorBound_beta}
Consider $\widehat{\beta}$ given by (\ref{est-1}). For any fixed pair of realizations of $X$ and $E$ , assume the following:
\begin{itemize}
\item[] \textbf{A1.} $\widehat{\Theta}_\epsilon$ is a deterministic matrix satisfying the bound: $\|\widehat{\Theta}_\epsilon-\Theta_\epsilon^*\|_\infty\leq\nu_\Theta$ where $\nu_\Theta= \eta_\Theta\left(\sqrt{\frac{\log p_2}{n}}\right)$ and $\eta_\Theta$ is some constant depending only on $\Theta_\epsilon^*$;
\item[] \textbf{A2.} $\widehat{\Gamma}\sim RE(\varphi,\phi)$, with $s^{**}\phi\leq \varphi/32$;
\item[] \textbf{A3.} $(\widehat{\Gamma},\widehat{\gamma})$ satisfies the deviation bound:
\begin{equation*}
\|\widehat{\gamma}-\widehat{\Gamma}\beta^*\|_\infty \leq \mathbb{Q}({\nu_\Theta})\sqrt{\frac{\log(p_1p_2)}{n}},
\end{equation*}
where $\mathbb{Q}({\nu_\Theta})$ is some quantity depending on $\nu_\Theta$. 
\end{itemize}
Then, for any $\lambda_n\geq 4\mathbb{Q}(\nu_\Theta)\sqrt{\frac{\log(p_1p_2)}{n}}$, the following bound holds:
\begin{equation*}
\|\widehat{\beta}-\beta^*\|_1 \leq 64s^{**}\lambda_n/\varphi. 
\end{equation*}
\end{theorem}

\begin{proof}
The statement of the theorem is a variation of Proposition 4.1 in \citet{basu2015estimation}, and its proof follows directly from the proof of the
 proposition in \citet[Appendix B]{basu2015estimation}. We only outline how the statement differs. In the original statement of Proposition 4.1 in \citet{basu2015estimation}, the authors provide the error bound for $\bar{\beta}$, obtained as per (\ref{est-1.1}) whose dimension is $qp^2$ with $q$ denoting the true lag of the vector-autoregressive process, under an RE condition for $\bar{\Gamma}$ and a deviation bound for $(\bar{\gamma},\bar{\Gamma})$. For our problem, we impose a similar RE condition on $\widehat{\Gamma}$ and deviation bound on $(\widehat{\gamma},\widehat{\Gamma})$, so as to yield a bound on $\widehat{\beta}$ that lies in a $p_1p_2$-dimensional space. 
\end{proof}

The following two propositions verify the RE condition for $\widehat{\Gamma}$ and deviation bound for $(\widehat{\Gamma},\widehat{\gamma})$ hold with high probability for a random pair $(X,E)$, given any symmetric, matrix $\widehat{\Theta}_\epsilon$ satisfying (A1). The proofs for these two propositions are given in the Appendix.

%
%
%
%

\begin{proposition}[Verification of RE condition for random $X$ and $E$]\label{prop:REcondition}
Consider any deterministic matrix $\widehat{\Theta}_{\epsilon}$ satisfying (A1). Let the sample size satisfy $n\succsim \max\{s^{**}\log p_1,d^2\log p_2\}$. With probability at least $1-2\exp(-c_3n)$ for some constant $c_3>0$, $\widehat{\Gamma}$ satisfies the following RE condition:
\begin{equation*}
\widehat{\Gamma} \equiv \widehat{\Theta}_\epsilon\otimes (X'X/n)\sim RE\left(\varphi^*(\min\limits_i\psi^i-d\nu_\Theta), \phi^* \max\limits_i(\psi^i+d\nu_\Theta)\right)
\end{equation*}
where $\varphi^*=\frac{\Lambda_{\min}(\Sigma^*_X)}{2}$, $\phi^*=(\varphi^*\log p_1)/n$, and $\psi^i$ is defined as:
\begin{equation*}
\psi^i:=\sigma^{ii}_\epsilon-\sum_{j\neq i}^{p_2}\sigma^{ij}_\epsilon,
\end{equation*}
where $\sigma^{ij}_\epsilon$'s denote the entries in $\Theta^*_\epsilon$ hence $\psi^i$ is the gap between its diagonal entry and the sum of off-diagonal entries for row $i$.
\end{proposition}

%
%
%
%

\begin{proposition}[Deviation bound for $(\widehat{\Gamma},\widehat{\gamma})$ for random $X$ and $E$] \label{prop:deviation}
Consider any deterministic matrix $\widehat{\Theta}_{\epsilon}$ satisfying (A1). Let sample size $n$ satisfy $n \succsim \log (p_1p_2)$. With probability at least 
\begin{equation*}
1-12c_1\exp[-(c_2^2-1)\log (p_1p_2)]\quad \text{for some } c_1>0,c_2>1
\end{equation*}
the following bound holds:
\begin{equation*}
\|\widehat{\gamma}-\widehat{\Gamma}\beta^*\|_\infty=\frac{1}{n}\left\|X'E\widehat{\Theta}_\epsilon\right\|_\infty  \leq \mathbb{Q}(\nu_\Theta)\sqrt{\frac{\log (p_1p_2)}{n}},  
\end{equation*}
where
\begin{equation}\label{Q-expression}
\mathbb{Q}(\nu_\Theta)= c_2\left\{ d\nu_\Theta\left[\Lambda_{\max}(\Sigma^*_X)\Lambda_{\max}(\Sigma^*_\epsilon)\right]^{1/2} + \left[ \frac{\Lambda_{\max}(\Sigma_X^*)}{\Lambda_{\min}(\Sigma_\epsilon^*)}\right]^{1/2} \right\} .
\end{equation}
\end{proposition}

\begin{remark} 
In Proposition~\ref{prop:REcondition}, the quantity $d^2\log p_2$ that shows up in the sample size requirement is a result of $\nu_\Theta =O(\sqrt{\log p_2/n})$, which is the common order of error in a generic graphical Lasso problem. Hence here we explicitly list it for the purpose of showing results for the generic graphical Lasso estimation problem. In our iterative algorithm, the order of $\nu_\Theta^{(k)}$ depends on the relative order of $p_1$ and $p_2$, which may potentially make the sample size requirement more stringent. This will be discussed in more detail in the proof of Theorem~\ref{thm:beta-theta-bound}.
\end{remark}

%
%
%
%

\medskip
Given the results in Theorem~\ref{thm:ErrorBound_beta}, Proposition~\ref{prop:REcondition} and Proposition~\ref{prop:deviation}, next we provide Corollary~\ref{cor:ErrorBoundB}, which gives the error bound for $\widehat{\beta}$ for random realizations of $X$ and $E$. 
Its proof is given in the Appendix.

\begin{corollary}[Error Bound for $\widehat{\beta}$ for random $X$ and $E$] \label{cor:ErrorBoundB}
Consider any determinisitic $\widehat{\Theta}_\epsilon$ satisfying the following elementwise $\ell_\infty$-bound:
\begin{equation*}
\|\widehat{\Theta}_\epsilon-\Theta_\epsilon^*\|_\infty \leq  \nu_\Theta,  \label{Theta-condition}
\end{equation*}
with $\nu_\Theta= \eta_\Theta\sqrt{\frac{\log p_2}{n}}$. Then for sample size $n \succsim \log(p_1p_2)$ and for any regularization parameter $\lambda_n\geq 4\mathbb{Q}(\nu_\Theta)\sqrt{\frac{\log (p_1p_2)}{n}}$ with the expression of $\mathbb{Q}(\cdot)$ given in (\ref{Q-expression}), there exists $c_1>0$ and $c_2>1$ such that with probability at least:
\begin{equation*}
1-12c_1\exp[-(c_2^2-1)\log (p_1p_2)] -  2\exp(-c_3n),
\end{equation*}
the following bound holds:
\begin{equation}\label{bbound-1}
\|\widehat{\beta}-\beta^*\|_1  \leq  64s^{**}\lambda_n/\varphi,
\end{equation}
where $\varphi = \frac{1}{2}\Lambda_{\min}(\Sigma_\epsilon^*)(\min\limits_i\psi^i-d\nu_\Theta)$.
\end{corollary}


%
%
%
%

\begin{theorem}[Error bound for $\widehat{\Theta}_\epsilon$ for fixed realizations of $X$ and $E$]\label{thm:ErrorBound_Theta}
Consider $\widehat{\Theta}_\epsilon$ given by (\ref{est-2}). For any fixed pair of realization $(X,E)$, assume the following:
\begin{itemize}
\item[] \textbf{B1.} $\widehat{\beta}$ is a deterministic vector satisfying $\|\widehat{\beta}-\beta^*\|_1\leq \nu_\beta$, where $\nu_\beta = \eta_{\beta}\left( \sqrt{\frac{\log(p_1p_2)}{n}}\right)$, with$\eta_\beta$ being some constant depending only on $\beta^*$;
\item[] \textbf{B2.} $\|\widehat{S}-\Sigma^*_\epsilon\|_\infty \leq g({\nu_\beta})$ where 
\begin{equation*}
\widehat{S} = \frac{1}{n}(Y-X\widehat{B})'(Y-X\widehat{B}),
\end{equation*}
and $g({\nu_\beta})$ is some quantity depending on $\nu_\beta$;
\item[] \textbf{B3.} Incoherence condition holds for $\Theta_\epsilon^*$. 
\end{itemize}
Then, for $\rho_n = (8/\xi)g(\nu_\beta)$ and sample size $n$ satisfying $n \succsim \log(p_1p_2)$,
the following error bound for $\widehat{\Theta}_\epsilon$ holds:
\begin{equation}\label{Thetabound}
\|\widehat{\Theta}_\epsilon-\Theta^*_\epsilon\|_\infty\leq \{2(1+8\xi^{-1})\kappa_{H^*}\} g(\nu_\beta),
\end{equation}
where $\xi$ is the incoherence parameter as defiend in Definition~\ref{def:incoherence}.
\end{theorem}

\begin{proof}
The statement of this theorem is a variation of Theorem~1 in \citet{ravikumar2011high}, so here, instead of providing a complete proof of the theorem, we only outline how the estimation problem differs in our setting, as well as the required changes in its proof.

In \citet{ravikumar2011high}, the authors consider the optimization problem in (\ref{est-2.1}), and show that for a random realization, with certain sample size requirement and choice of the regularization parameter, the following bound for $\bar{\Theta}_\epsilon$ holds with probability at least $1-1/p_2^\tau$ for some $\tau>2$:
\begin{equation}\label{glasso-bound}
\|\bar{\Theta}_\epsilon -\Theta_\epsilon^*\|_\infty \leq \{2(1+8\xi^{-1})\kappa_{H^*}\}\bar{\delta}_f(p_2^\tau,n),
\end{equation}
where $\bar{\delta}(r,n)$ is defined as:
\begin{equation}\label{delta-bar}
\bar{\delta}(r,n): = 8(1+4\sigma^2)\max_i(\Sigma^*_{\epsilon,ii})\sqrt{\frac{2\log(4r)}{n}}.
\end{equation}
The quantity  $\bar{\delta}(p_2^\tau,n)$ that shows up in expression (\ref{glasso-bound}) is the bound for $\|S-\Sigma_\epsilon^*\|_\infty \equiv \|\widehat{\Sigma}_\epsilon - \Sigma_\epsilon^*\|_\infty$. In particular, in Lemma~8 \citep{ravikumar2011high}, they show that with probability at least $1-1/p_2^\tau$, $\tau>2$,  the following bound holds:
\begin{equation*}
\|S - \Sigma_\epsilon^*\|_\infty \leq \bar{\delta}(p_2^\tau,n).
\end{equation*}
In our optimization problem (\ref{est-2}), we are using $\widehat{S}$ instead of $S$, hence a bound for $\|\widehat{S}-\Sigma_\epsilon^*\|_\infty$ is necessary, and the remaining argument in the proof of Theorem~1 \citep{ravikumar2011high} will follow through. 

Therefore in our theorem statement, we use $g(\nu_\beta)$ as a bound for $\|\widehat{S}-\Sigma_\epsilon^*\|_\infty$ then yield the bound for $\|\widehat{\Theta}_\epsilon-\Theta^*_\epsilon\|_\infty$, since we are using the surrogate error $\widehat{E}=Y-X\widehat{B}$ in the estimation, instead of the true error $E$. 
\end{proof}

Proposition~\ref{prop:residual-concentration} gives an explicit expression for $g(\nu_\beta)$ under condition (B1). Specifically, it shows how well $\widehat{S}$ concentrates around $\Sigma^*_\epsilon$ for random $X$ and $E$, given some small-errored $\widehat{B}$ (or $\widehat{\beta}$, equivalently), and its proof is given in the appendix.

%
%
%
%

\begin{proposition}\label{prop:residual-concentration}
Consider any determinisitc $\widehat{\beta}$ satisfying (B1). Then for sample size $n$ satisfying $n \succsim \log(p_1p_2)$, 
%
with probability at least:
\begin{equation*}
1-1/p_1^{\tau_1-2}-1/p_2^{\tau_2-2}-6c_1\exp[-(c_2^2-1)\log (p_1p_2)], \quad \text{for some }c_1>0,c_2>1,\tau_1,\tau_2>2,
\end{equation*}
the following bound holds: 
\begin{equation}\label{bound:S-hat}
\begin{split}
\|\widehat{S} - \Sigma^*_\epsilon\|_\infty &\leq g(\nu_\beta),
\end{split} 
\end{equation}
where
\begin{equation}\label{g-expression}
\begin{split}
g(\nu_\beta) & = \sqrt{\frac{\log 4 + \tau_2 \log p_2}{c^*_\epsilon n}} + \nu_\beta^2 \left( \sqrt{\frac{\log 4 + \tau_1 \log p_1}{c^*_Xn}} +  \max_i(\Sigma^*_{X,ii})\right)\\
& \quad + 2c_2\nu_\beta\left[\Lambda_{\max}(\Sigma^*_X)\Lambda_{\max}(\Sigma^*_\epsilon)\right]^{1/2} \sqrt{\frac{\log (p_1p_2) }{n}},
\end{split}
\end{equation}
$c_\epsilon^*$ and $c_X^*$ are population quantities given in (\ref{exp:c_epsilon}) and (\ref{exp:c_X}), respectively.
\end{proposition} 

%
%
%
%

Given Theorem~\ref{thm:ErrorBound_Theta} and Proposition~\ref{prop:residual-concentration}, we provide Corollary~\ref{cor:erorrboundTheta}, which gives the error bound for $\widehat{\Theta}_\epsilon$ for random realizations of $X$ and $E$:

\begin{corollary}[Error bound for $\widehat{\Theta}$ for random $X$ and $E$] \label{cor:erorrboundTheta}
Consider any deterministic $\widehat{\beta}$ satisfying the following bound:
\begin{equation*}
\|\widehat{\beta} - \beta^*\|_1\leq \nu_\beta
\end{equation*}
with $\nu_\beta = \eta_\beta\sqrt{\frac{\log (p_1p_2)}{n}}$. Also suppose the incoherence condition (B3) is satisfied. Then, for sample size $n \succsim \log(p_1p_2) $  and regularization parameter $\rho_n=(8/\xi)g(\nu_\beta)$ with $g(\nu_\beta)$ given in (\ref{g-expression}), with probability at least
\begin{equation*}
1-1/p_1^{\tau_1-2}-1/p_2^{\tau_2-2}-6c_1\exp[-(c_2^2-1)\log (p_1p_2)], \quad \text{for some }c_1>0,c_2>1,\tau_1,\tau_2>2,
\end{equation*}
the following bound holds:
\begin{equation*}
\|\widehat{\Theta}_\epsilon-\Theta^*_\epsilon\|_\infty \leq \{2(1+8\xi^{-1})\kappa_{H^*}\} g(\nu_\beta).
\end{equation*}
\end{corollary}

%
%
%
%

\medskip
After providing the error bound for (\ref{est-1}) and (\ref{est-2}), in Theorem~\ref{thm:beta-theta-bound} we establish that with high probability, the error of the sequence of estimates obtained in the alternating search step of the algorithm described in Section~\ref{sec:estimation} is {\em uniformly} bounded; that is, the sequence of estimates lie in a non-expanding ball around the true value of the parameters uniformly with a radius that doesn't depend on $k$, the iteration number. 

\begin{theorem}[Error bound for $\{\widehat{\beta}^{(k)}\}$ and $\{\widehat{\Theta}_\epsilon^{(k)}\}$]\label{thm:beta-theta-bound}
Consider the iterative algorithm given in Section~\ref{sec:estimation} that gives rise to sequences of $\{\widehat{\beta}^{(k)}\}$ and $\{\widehat{\Theta}_\epsilon^{(k)}\}$ alternately. For random realization of $X$ and $E$, we assume the following:
\begin{itemize}
\item[] \textbf{C1.} The incoherence condition holds for $\Theta^*_\epsilon$.
\item[] \textbf{C2.} $\Theta^*_\epsilon$ is diagonally dominant.
\item[] \textbf{C3.} The maximum sparsity level for all $p_2$ regression $s^*$ satisfies $s^*=o(n/\log p_1)$.
\end{itemize}
(I) For sample size satisfying $n\succsim \log(p_1p_2)$, there exist constants $c_1>0,c_2>1,c_3>0$ such that for any
\begin{equation*}
\lambda_n^0\geq 4c_2\left[\Lambda_{\max}(\Sigma_X^*)\Lambda_{\max}(\Sigma_\epsilon^*)\right]^{1/2}\sqrt{\frac{\log(p_1p_2)}{n}},
\end{equation*}
with probability at least $1-2\exp(-c_3n)-6c_1\exp[-(c_2^2-1)\log(p_1p_2)]$, the initial estimate $\widehat{\beta}^{(0)}\equiv vec(\widehat{B}^{(0)})$ satisfies the following bound:  
\begin{equation}\label{eqn:beta0bound}
\|\widehat{\beta}^{(0)} - \beta^*\|_1 \leq 64 s^{**}\lambda_n^0/\varphi^*\equiv \nu_\beta^{(0)},
\end{equation}
where $\varphi^* = \Lambda_{\min}(\Sigma_X^*)/2$. Moreover, by choosing $\rho_n^{0} = (\frac{8}{\xi})g(\nu_\beta^{(0)})$ where the expression for $g(\cdot)$ is given in (\ref{g-expression}), with probability at least $$ 1-1/p_1^{\tau_1-2}-1/p_2^{\tau_2-2}-2\exp(-c_3n) - 6c_1\exp[-(c_2^2-1)\log (p_1p_2)],\quad \text{for some }\tau_1,\tau_2>2$$ 
the following bound holds:
\begin{equation}\label{eqn:Theta0bound}
\|\widehat{\Theta}^{(0)}_\epsilon - \Theta^*_\epsilon\|_\infty \leq \{2(1+8\xi^{-1})\kappa_{H^*}\}g({\nu_\beta^{(0)}}) \equiv \nu_\Theta^{(0)}.
\end{equation}
(II) For sample size satisfying $n \succsim d^2\log(p_1p_2) $, for any iteration $k\geq 1$, with probability at least 
\begin{equation*}
1-1/p_1^{\tau_1-2} - 1/p_2^{\tau_2-2} - 12c_1\exp[-(c_2^2-1)\log (p_1p_2)] - 2\exp[-c_3n], 
\end{equation*}
the following bounds hold for all $\widehat{\beta}^{(k)}$ and $\widehat{\Theta}^{(k)}_\epsilon$: 
\begin{eqnarray*}
\|\widehat{\beta}^{(k)} - \beta^*\|_1 &\leq & C_\beta\left( s^{**}\sqrt{\frac{\log(p_1p_2)}{n}}\right),\\
\|\widehat{\Theta}^{(k)}_\epsilon - \Theta^*_\epsilon\|_\infty &\leq& C_\Theta\left(\sqrt{\frac{\log(p_1p_2)}{n}}\right).
\end{eqnarray*}
where $s^{**}$ is the sparsity of $\beta^*$, $C_\beta$ and $C_\Theta$ are constants depending only on $\beta^*$ and $\Theta_\epsilon^*$, respectively.
\end{theorem}

\medskip
\begin{proof} We first consider part (I) of the theorem. Note that by (\ref{eqn:initialB}), $\widehat{\beta}^{(0)}$ can be equivalently written as:
\begin{equation}\label{eqn:beta0}
\widehat{\beta}^{(0)} \equiv \argmin\limits_{\beta\in\mathbb{R}^{p_1\times p_2}}\left\{-2\beta'\gamma^0 + \beta'\Gamma^0\beta + \lambda_n^0\|\beta\|_1 \right\},
\end{equation}
where 
\begin{equation*}
\Gamma^{(0)} = \mathrm{I} \otimes \frac{X'X}{n}, \quad \gamma^{(0)} = (\mathrm{I} \otimes X')\mathrm{vec}{Y}/n.
\end{equation*}
Consider the following events:
\begin{itemize}
\item[\bf E1.] $\left\{\frac{X'X}{n}\sim RE(\varphi^*,\phi^*)\right\}$, 
\item[\bf E2.] $\left\{\frac{1}{n}\left\| X'E \right\|_\infty \leq c_2\left[\Lambda_{\max}(\Sigma_X^*)\Lambda_{\max}(\Sigma_\epsilon^*)\right]^{1/2}\sqrt{\frac{\log(p_1p_2)}{n}}\right\}$.
\end{itemize}
Note that \textbf{E1} $\cap$ \textbf{E2} implies the following events:
\begin{equation*}
\Gamma^{(0)} \equiv \mathrm{I}\otimes \frac{X'X}{n} \sim RE(\varphi^*,\phi^*), \quad \text{where }\varphi^*=\Lambda_{\min}(\Sigma_X^*)/2. 
\end{equation*}
and
\begin{equation}\label{bound:gamma0}
\|\gamma^{(0)} - \Gamma^{(0)}\beta^*\|_\infty = \frac{1}{n}\left\| X'E \right\|_\infty \leq c_2\left[\Lambda_{\max}(\Sigma_X^*)\Lambda_{\max}(\Sigma_\epsilon^*)\right]^{1/2}\sqrt{\frac{\log(p_1p_2)}{n}}.
\end{equation}
Hence, by Proposition~4.1 of \citet{basu2015estimation}, the bound \eqref{eqn:beta0bound} holds on \textbf{E1} $\cap$ \textbf{E2}.

\noindent By Lemmas~\ref{lemma:restateB.1} and ~\ref{lemma:RESX}, $\mathbb{P}(\mathbf{E1})$ is at least $1-2\exp(-c_3n)$, for some $c_3>0$. 
By Lemma~\ref{lemma:deviation_aux}, $\mathbb{P}(\mathbf{E2})$  is at least $1-6c_1\exp[-(c_2^2-1)\log(p_1p_2)]$ for some $c_1>0$, $c_2>1$.
Hence, with probability at least 
\begin{equation*}
\mathbb{P} \left(\mathbf{E1} \cap \mathbf{E2} \right) \ge 1 - \mathbb{P} \left(\mathbf{E1}^c \right) - \mathbb{P} \left(\mathbf{E2}^c \right)
\end{equation*}
the bound in (\ref{eqn:beta0bound}) holds, which proves the first part of (I). In particular, we have $\|\hat{\beta}^0 - \beta^*\|_1 \le \nu_\beta^{(0)}\sim O(\sqrt{\log(p_1p_2)/n})$ on $\mathbf{E1} \cap \mathbf{E2}$. 

\noindent To prove the second part of (I), note that by Theorem \ref{thm:ErrorBound_Theta} the bound in (\ref{eqn:Theta0bound}) holds when B1-B3 are satisfied. Now, from the argument above, B1 holds on the event \textbf{E1} $\cap$ \textbf{E2}. Also, from the proof of Proposition~\ref{prop:residual-concentration}, B2 is satisfied, i.e., 
\begin{equation}\label{eqn:S0prob}
\left\|\widehat{S}^{(0)} - \Sigma_\epsilon^*\right\|_\infty\leq g(\nu_\beta^{(0)}),\quad \text{where }\widehat{S}^{(0)}= \frac{1}{n}(Y-X\widehat{B}^{(0)})'(Y-X\widehat{B}^{(0)}),
\end{equation}
on \textbf{E1} $\cap$ \textbf{E2} $\cap$ \textbf{E3} $\cap$ \textbf{E4}, where the events \textbf{E3} and \textbf{E4} are given by:
\begin{itemize}
\item[\bf E3.] $\left\{  \left\|\frac{E'E}{n}-\Sigma^*_\epsilon\right\|_\infty \leq \sqrt{\frac{\log 4 + \tau_2 \log p_2}{c_\epsilon^* n}} \right\}$ for some $\tau_2>2$ and $c_\epsilon^*>0$ that depends on $\Sigma_{\epsilon}^*$,
\item[\bf E4.] $\left\{  \left\|\frac{X'X}{n}-\Sigma^*_X\right\|_\infty \leq \sqrt{\frac{\log 4 + \tau_1 \log p_1}{c_X^* n}} \right\}$ for some $\tau_1>2$ and $c_X*>0$ that depends on $\Sigma_{X}^*$.
\end{itemize}

Therefore, the probability of the bound for $\widehat{\Theta}_\epsilon^{(0)}$ in (\ref{eqn:Theta0bound}) to hold is at least
\begin{equation}\label{eqn:prob-intersection}
\mathbb{P}\left(\mathbf{E1}\cap \mathbf{E2} \cap \mathbf{E3} \cap \mathbf{E4}\right),
\end{equation}
By Lemma~\ref{lemma:RESX}, Lemma~\ref{lemma:deviation_aux} and the proof of Proposition~\ref{prop:residual-concentration}, the probability in (\ref{eqn:prob-intersection}) is lower bounded by:
\begin{equation*}
1-2\exp(-c_3n) - 6c_1\exp[-(c_2^2-1)\log (p_1p_2)]-1/p_1^{\tau_1-2}-1/p_2^{\tau_2-2}.
\end{equation*}

Consider the following two cases where the relative order of $p_1$ and $p_2$ differ. Case 1: $p_1\prec p_2$, then $\nu_\Theta^{(0)}\sim O(\sqrt{\log p_2/n})$; case 2: $p_1\succsim p_2$, then $\nu_\Theta^{(0)}\sim O\left(\log(p_1p_2)/n\right)$. In either case, since we are assuming $\log(p_1p_2)/n$ to be a small quantity and it follows that $\sqrt{\log(p_1p_2)/n}\succsim \log(p_1p_2)/n$, the following bound always holds: 
\begin{equation*}
\nu_\Theta^{(0)} \leq C_\Theta \sqrt{\frac{\log(p_1p_2)}{n}} \equiv  M_\Theta,
\end{equation*}
where $C_\Theta$ is some large fixed constant that bounds the constant terms in front of $\sqrt{\log(p_1p_2)/n}$.

Now we consider part (II) of the theorem. Note that for each $k\geq 1$, $\widehat{\beta}^{(k)}$ and $\widehat{\Theta}_\epsilon^{(k)}$ are obtained via solving the following two optimizations: 
\begin{eqnarray}
\widehat{\beta}^{(k)} & =& \argmin\limits_{\beta\in\mathbb{R}^{p_1\times p_2}} \left\{ -2\beta'\widehat{\gamma}^{(k-1)} + \beta'\widehat{\Gamma}^{(k-1)}\beta + \lambda_n\|\beta\|_1 \right\}, \\
\widehat{\Theta}_\epsilon^{(k)} & = & \argmin\limits_{\Theta_\epsilon\in\mathbb{S}_{++}^{p_2\times p_2}}\left\{ \log\det \Theta_\epsilon - \text{tr}(\widehat{S}^{(k)} \Theta_\epsilon) + \rho_n\|\Theta_\epsilon\|_{1,\text{off}}  \right\},
\end{eqnarray}
where 
\begin{equation*}
\widehat{\gamma}^{(k)} = \widehat{\Theta}^{(k)}\otimes \frac{X'Y}{n}, \quad \widehat{\Gamma}^{(k)} = \widehat{\Theta}^{(k)}\otimes \frac{X'X}{n}, \quad \widehat{S}^{(k)} = \frac{1}{n}(Y-X\widehat{B}^{(k)})'(Y-X\widehat{B}^{(k)}).
\end{equation*}
Consider the bound on $\hat{\beta}^{(k)}$ for $k=1$. The argument is similar to that of $\hat{\beta}^{(0)}$, with appropriate modifications to account for the fact that the objective function now involves log likelihood instead of least squares. Formally, we consider the event \textbf{E1} $\cap$ \textbf{E2} $\cap$ \textbf{E3} $\cap$ \textbf{E4} $\cap$ \textbf{E5}, where 
\begin{itemize}
\item[\bf E5.] $\left\{\frac{1}{n}\left\| X'E\Theta^*_\epsilon\right\|_\infty \leq c_2\left[\frac{\Lambda_{\max}(\Sigma^*_X)}{\Lambda_{\min}(\Sigma^*_\epsilon)}\right]^{1/2} \sqrt{\frac{\log (p_1p_2) }{n}}\right\}$.
\end{itemize}
Note that $\{\|\widehat{\Theta}^{(0)}_\epsilon-\Theta^*_\epsilon\|_\infty\leq \nu_\Theta^{(0)}\}$ holds on this event.  By Lemma~\ref{lemma:deviation_aux}, $\mathbb{P}(\textbf{E5}) \geq 1- 6c_1\exp[-(c_2^2-1)\log(p_1p_2)]$. Combining this with the lower bound on \eqref{eqn:prob-intersection} and 
the sample size requirement (note this sample size requirement can be relaxed to $n\succsim \log(p_1p_2)$ if $p_1\prec p_2$), we obtain that with probability at least
\begin{equation*}
1-1/p_1^{\tau_1-2} - 1/p_2^{\tau_2-2} - 12c_1\exp[-(c^2_2-1)\log (p_1p_2)] - 2\exp[-c_3n], 
\end{equation*}
the following three events hold simultaneously: 
\begin{itemize}
\item[] \textbf{\em A1'} $\|\widehat{\Theta}^{(0)}_\epsilon-\Theta^*_\epsilon\|_\infty\leq \nu_\Theta^{(0)}\precsim O(\sqrt{\log(p_1p_2)/n})$;
\item[] \textbf{\em A2'} $\widehat{\Gamma}^{(0)}\sim RE(\varphi^{(0)},\phi^{(0)})$ where
\begin{equation*}
\varphi^{(0)} \geq \frac{\Lambda_{\min}(\Sigma^*_X)}{2}(\min_i\psi^i - dM_\Theta)~~~\text{and} ~~~ \phi^{(0)} \leq \frac{\log p_1}{n}\frac{\Lambda_{\min}(\Sigma^*_X)}{2}(\max_j\psi^j + dM_\Theta);
\end{equation*}
\item[] \textbf{\em A3'} $\|\widehat{\gamma}^{(0)}-\widehat{\Gamma}^{(0)}\beta^*\|_\infty \leq \mathbb{Q}(\nu_\Theta^{(0)})\sqrt{\frac{\log(p_1p_2)}{n}}$ with the expression for $\mathbb{Q}(\cdot)$ given in (\ref{Q-expression}).
\end{itemize}
By Theorem~\ref{thm:ErrorBound_beta}, by choosing $\lambda_n\geq 4\mathbb{Q}(M_\Theta)\sqrt{\frac{\log(p_1p_2)}{n}}$, the following bound holds:
\begin{equation}\label{beta-bound}
\|\widehat{\beta}^{(1)}-\beta^*\|_1 \leq 64s^{**}\lambda_n/\varphi^{(0)}
\end{equation}

The error bound for $\widehat{\Theta}^{(1)}_\epsilon$ can now be established using the same argument for $\hat{\Theta}^{(0)}_\epsilon$, with the only difference that now we consider the event $\mathbf{E1} \cap \ldots \cap \mathbf{E5}$ instead of $\mathbf{E1} \cap \ldots \cap \mathbf{E4}$ and use \eqref{beta-bound} instead of \eqref{eqn:beta0bound}.

Note that an upper bound for the leading term of the right hand side of (\ref{beta-bound}) is at most of the order $O(\sqrt{\log(p_1p_2)/n})$, and can be written as:
\begin{equation*}
C_\beta\left( s^{**}\sqrt{\frac{\log(p_1p_2)}{n}}\right)\equiv M_\beta,
\end{equation*}
with $C_\beta$ being some potentially large number that bounds the constant term. Notice that $M_\beta$ is of the same order as $\nu_\beta^{(0)}$;
thus, for $\widehat{\Theta}^{(1)}_\epsilon$, we can also achieve the following bound: 
\begin{equation*}
\|\widehat{\Theta}_\epsilon^{(1)}-\Theta^*_\epsilon\|_\infty\leq M_\Theta
\end{equation*}
with high probability since we are assuming $C_\Theta$ to be some potentially large number. 

Note that the events $\textbf{E1}, \ldots, \textbf{E5}$ rely only on the parameters and not on the estimated quantities, and on their intersection we have uniform upper bounds on the errors of $\widehat{\beta}^{(k)}$ and $\widehat{\Theta}^{k}_{\epsilon}$ for $k = 0, 1$. Hence the error bounds for $k=1$ can be used to invoke Theorems 2 and 3 inductively on realizations $X$ and $E$ from the set $\textbf{E1} \cap \ldots \cap \textbf{E5}$ to provide high probability error bounds for all subsequent iterates as well. This leads to the uniform error bounds of part (II) with the desired probability.

\end{proof}


\medskip
As a direct result of Proposition~1 in \citet{basu2015estimation} and Corollary~3 in \citet{ravikumar2011high}, the following bound also holds:
\begin{corollary}\label{cor:beta-theta-bound}
Under the same set of conditions C1, C2 and C3 in Theorem~\ref{thm:beta-theta-bound}, there exists $\tau_1,\tau_2>2$, $c_1>0,c_2>1,c_3>0$ and constants $C_\beta'$ and $C_\Theta'$ such that for all iterations $k$, the following bound holds:
\begin{eqnarray*}
\|\widehat{\beta}^{(k)} - \beta^*\|_F &\leq & C'_\beta\left( \sqrt{\frac{s^{**}\log(p_1p_2)}{n}}\right),\\
\|\widehat{\Theta}^{(k)}_\epsilon - \Theta^*_\epsilon\|_F &\leq& C'_\Theta\sqrt{\frac{(s_\epsilon^*+p_2)\log(p_1p_2)}{n}},
\end{eqnarray*}
with probability at least 
\begin{equation*}
1-1/p_1^{\tau_1-2} - 1/p_2^{\tau_2-2} - 12c_1\exp[-(c_2^2-1)\log (p_1p_2)] - 2\exp[-c_3n],
\end{equation*}
where $s^{**}$ and $s^*_\epsilon$ are the sparsity for $\beta^*$ and $\Theta^*_\epsilon$, respectively.
\end{corollary}

\begin{remark}\label{remark:restriction}
As mentioned earlier in this subsection, the actual implementation of the alternating search step is restricted to a subspace of $\mathbb{R}^{p_1\times p_2}$. 
Next, we outline the corresponding theoretical results for this specific scenario in which for each regression $j$, some {\em fixed superset} of the indices of true covariates is given, and the regressions are restricted to these supersets, respectively. Note that we need to make sure that the restricted subspace contains all the true covariates for the results below to be valid. 

Let $S_j$ denote the given {\em fixed superset} for each regression $j$, and  we consider regressing the response on $X_{S_j}$. We use $\widehat{\beta}_{\mathrm{R}}^{(k)}$ to denote the corresponding vectorized estimator of iteration $k$, that is, 
\begin{equation*}
\widehat{\beta}^{(k)}_{\mathrm{R}} = (\widehat{B}^{(k)'}_{1,\mathrm{Restricted}},\cdots,\widehat{B}^{(k)'}_{p_2,\mathrm{Restricted}})'
\end{equation*}
where $\widehat{B}^{(k)'}_{j,\mathrm{Restricted}}$ is obtained by doing the regression in (\ref{eqn:Bjunrestricted}), however with the indices of covariates restricted to $S_j$. Also, we let $\beta^*_{\mathrm{R}}$ be the corresponding true value of $\widehat{\beta}^{(k)}_{\mathrm{R}}$. Note that always holds that 
\begin{equation*}
\|\widehat{\beta}^{(k)}_{\mathrm{R}}-\beta^*_{\mathrm{R}}\|= \ \|\widehat{\beta}^{(k)}-\beta^*\|.
\end{equation*} 

Now let   
\begin{equation*}
\bar{S} = \bigcup\limits_{j\in\{1,\cdots,p_2\}} S_j
\end{equation*}
and let $\bar{s}$ be its cardinality. It can be shown that the best achievable error bound for $\widehat{\beta}^{(k)}_{\mathrm{R}}$ is identical to $\widehat{\beta}^{(k)}_{\bar{S}}$, where $\widehat{\beta}^{(k)}_{\bar{S}}$ is obtained by considering covariates $X_{\bar{S}}$ for all $p_2$ regressions, instead of the entire $X$. For this specific reason, formally, we state the theoretical results for the case where we consider regressing on $X_{\bar{S}}$, which is almost identical to the generic case.

Suppose conditions C1, C2 and C3 in Theorem 4 hold, then there exists constants $c_1>0,c_2>1,c_3>0,\tau_1>2,\tau_2>2$ such that: (I) for sample size satisfying $n\succsim \log(\bar{s}p_2)$, w.p. at least $1-2\exp(-c_3n)-6c_1\exp[-(c_2^2-1)\log(\bar{s}p_2)]$, for any $$
\lambda_n^0\geq 4c_2\left[\Lambda_{\max}(\Sigma_{X_{\bar{S}}}^*)\Lambda_{\max}(\Sigma_\epsilon^*)\right]^{1/2}\sqrt{\frac{\log(\bar{s}p_2)}{n}},
$$ the initial estimate $\widehat{\beta}^{(0)}_{\bar{S}}$ satisfies the following bound:  
\begin{equation*}
\|\widehat{\beta}^{(0)}_{\bar{S}} - \beta^*_{\bar{S}}\|_1 \leq 64 s^{**}\lambda_n^0/\varphi^*_{\bar{S}}\equiv \nu_{\beta_{\bar{S}}}^{(0)},
\end{equation*}
where $\varphi^*_{\bar{S}} = \Lambda_{\min}(\Sigma_{X_{\bar{S}}}^*)/2$. Moreover, by choosing $\rho_n^{0} = (\frac{8}{\xi})g(\nu_{\beta_{\bar{S}}}^{(0)})$ where the expression for $g(\cdot)$ is given in (\ref{g-expression}), with probability at least $$ 1-1/\bar{s}^{\tau_1-2}-1/p_2^{\tau_2-2}-2\exp(-c_3n) - 6c_1\exp[-(c_2^2-1)\log (\bar{s}p_2)],$$ 
the following bound holds:
\begin{equation*}
\|\widehat{\Theta}^{(0)}_\epsilon - \Theta^*_\epsilon\|_\infty \leq \{2(1+8\xi^{-1})\kappa_{H^*}\}g({\nu_{\beta_{\bar{S}}}^{(0)}}) \equiv \nu_\Theta^{(0)}.
\end{equation*}
(II) For sample size satisfying $n \succsim d^2\log(\bar{s}p_2) $, for any iteration $k\geq 1$, with probability at least 
\begin{equation*}
1-1/\bar{s}^{\tau_1-2} - 1/p_2^{\tau_2-2} - 12c_1\exp[-(c_2^2>1)\log (\bar{s}p_2)] - 2\exp[-c_3n], 
\end{equation*}
the following bound hold for all $\widehat{\beta}_{\bar{S}}^{(k)}$ and $\widehat{\Theta}^{(k)}_\epsilon$: 
\begin{eqnarray*}
&\|\widehat{\beta}_{\bar{S}}^{(k)} - \beta^*\|_1 \leq  C_\beta\left( s^{**}\sqrt{\frac{\log(\bar{s}p_2)}{n}}\right), &\|\widehat{\beta}^{(k)}_{\bar{S}} - \beta^*\|_F \leq  C'_\beta\left( \sqrt{\frac{s^{**}\log(\bar{s}p_2)}{n}}\right) \\
&\|\widehat{\Theta}^{(k)}_\epsilon - \Theta^*_\epsilon\|_\infty \leq C_\Theta\left(\sqrt{\frac{\log(\bar{s}p_2)}{n}}\right), 
&\|\widehat{\Theta}^{(k)}_\epsilon - \Theta^*_\epsilon\|_F \leq C'_\Theta\sqrt{\frac{(s_\epsilon^*+p_2)\log(\bar{s}p_2)}{n}}
\end{eqnarray*}
where $s^{**}$ is the sparsity of $\beta^*$, $C_\beta$, $C_{\beta}'$, $C_\Theta$ and $C_{\theta}'$ are all constants that do not depend on $n,\bar{S},p_2$. 
\end{remark}

%
%

\subsection{\normalsize Family-Wise Error Rate  control of the Screening Step}\label{sec:FWER}

As mentioned in the Introduction, for the iterative algorithm to work effectively, it is crucial to initialize from points that are close to the true parameters. Our screening step provides such guarantees {\em asymptotically}. Based on the screening step described in Section~\ref{sec:estimation}, initial esimates for each column of the regression matrix are obtained by Lasso or Ridge regression with the support set restricted to the one identified by the screening step. It is desirable for the screening step to correctly identify the true support set. In particular, we would like to retain as many true positive predictor variables as possible without discovering too many false positive ones. The following theorem states that as long as $\log(p_1p_2)/n = o(1)$ and the sparsity is not beyond a specified level, the screening step will be able to recover all true positive predictors, while keeping the family-wise type I error under control. 

\begin{theorem}
Let $S^*_j$ denote the true support set of the $j$th regression and $s^*_j$ be its cardinality. Suppose that $\log(p_1p_2)/n\rightarrow 0$ and the following condition for sparsity holds:
\begin{equation*}
\max\{s^*_j,j=1,\cdots,p_2\} = o(\sqrt{n}/\log p_1).
\end{equation*}
Then, the screening step described in Section 2.2 will correctly recover $S^*_j$ for all $j=1,\cdots,p_2$ with probability approaching to 1, while keeping the family-wise type I error rate under the prespecified level $\alpha$. 
\end{theorem}

\begin{proof}
First, we note that with a Bonferroni correction, the family-wise type I error will be automatically controlled at level $\alpha$. Hence, we will focus on the power of the screening step. Also, from Theorem 7 of \citet{javanmard2014confidence}, it is easy to see that all the arguments below hold for a large set of random realizations of $X$, whose probability approaches 1 under the specified asymptotic regime when the eigenvalues of $\Sigma_X$ are bounded away from $0$ and infinity.

Let $B^*=\begin{bmatrix}
B_1^* & \cdots & B_{p_2}^*
\end{bmatrix}$ denote the true value of the regression coefficients and $\check{B}_j,j=1,\cdots,p_2$ denote the estimates given by the de-biased Lasso procedure in \citet{javanmard2014confidence}. With the given level for sparsity, by Theorem~8 in  \citet{javanmard2014confidence}, each $\check{B}_j$ satisfies the following:
\begin{equation*}
\sqrt{n}(\check{B}_j-B_j^*) \sim \mathcal{N}\left(0,\sigma^2 M_j\widehat{\Sigma}_XM_j'\right),
\end{equation*}
where $\widehat{\Sigma}_X$ is the sample covariance matrix of the predictors $X$, $\sigma$ is the population noise level of the error term $\epsilon_j$, and $M_j$ is the matrix corresponding to the $j$th regression, produced by the procedure described in \citet{javanmard2014confidence}\footnote{Details of the procedure is described in p.2871 in \citet{javanmard2014confidence}, with $M$ being an intermediate quantity obtained by solving an optimization problem.}. Let $\check{B}_{j,i}$ denote the $i$th coordinate of the $j$th regression coefficient vector $\check{B}_j$ and $\check{\Sigma}_{j}$ be the covariance matrix of the estimator $\check{B}_j$, then 
\begin{equation*}
\check{\Sigma}_j = \frac{\sigma^2}{n}M_j\widehat{\Sigma}_XM_j',
\end{equation*}
and in particular, the variance of $\check{B}_{j,i}$ is $\check{\Sigma}_{j,ii}:=\check{\sigma}^j_{ii}$. Using these notations, for a prespecified level $\alpha$, the test statistics for testing $H^{ji}_0:B^*_{j,i}=0$ vs. $H^{ji}_A:B^*_{j,i}\neq 0$, for all $i=1,\cdots,p_1;j=1,\cdots,p_2$ can be equivalently written as: 
\begin{equation*}
\widehat{T}_{j,i} = \begin{cases}
1 \quad &\text{if }|\check{B}_{j,i}|/\check{\sigma}^j_{ii}> z_{\alpha/(2p_1p_2)}, \\
0 \quad &\text{otherwise.}
\end{cases}
\end{equation*}
where $z_{\alpha}$ denotes the upper $\alpha$ quantiles of $\mathcal{N}(0,1)$. 

Define the ``family-wise" power as follows: 
\begin{equation*}
\begin{split}
\mathbb{P}\left( \text{all true alternatives are detected} \right) & = \mathbb{P}\left( \bigcap\limits_{1\leq j\leq p_2}\bigcap\limits_{k\in S^*_j}\{\widehat{T}_{j,k}=1\}\right) \\
& = 1 - \mathbb{P}\left( \bigcup\limits_{1\leq j\leq p_2}\bigcup\limits_{k\in S^*_j}\{\widehat{T}_{j,k}=0\}\right).
\end{split}
\end{equation*}
Correspondingly, the family-wise type II error can be  written as:
\begin{equation}\label{type2}
\mathbb{P}\left( \bigcup\limits_{1\leq j\leq p_2}\bigcup\limits_{k\in S^*_j}\{\widehat{T}_{j,k}=0\}\right) \leq \sum_{j=1}^{p_2}\sum_{k\in S^*_j} \mathbb{P}\left( \widehat{T}_{j,k}=0\right). 
\end{equation}
By Theorem~16 in \citet{javanmard2014confidence}, asymptotically, $\forall k\in S_j,j=1,\cdots,p_2$:
\begin{equation}\label{powerbound}
\begin{split}
\mathbb{P}\left(\widehat{T}_{j,k}=0\right) & \leq 1 - G\left(\frac{\alpha}{p_1p_2},\frac{\sqrt{n}\gamma}{\sigma[\Sigma^{-1}_{k,k}]^{1/2}}\right); \qquad 0<\gamma\leq \min|B^*_{j,k}|,~~\forall k\in S_j,j=1,\cdots,p_2.
\end{split}
\end{equation}
Here 
\begin{equation*}
G(\alpha,u) \equiv  2- \mathbb{P}(\Phi < z_{\alpha/2}+u) - \mathbb{P}(\Phi < z_{\alpha/2} - u),
\end{equation*}
where we use $\Phi$ to denote the random variable following a standard Gaussian distribution. Hence, (\ref{powerbound}) can be re-written as:
\begin{equation}\label{type2:ind}
\begin{split}
\mathbb{P}\left(\widehat{T}_{j,k}=0\right) & \leq 1 - G\left(\frac{\alpha}{p_1p_2},\frac{\sqrt{n}\gamma}{\sigma[\Sigma^{-1}_{k,k}]^{1/2}}\right)  \\
& =   \mathbb{P}\left( \Phi < z_{\alpha/(2p_1p_2)} - \frac{\sqrt{n}\gamma}{\sigma[\Sigma^{-1}_{k,k}]^{1/2}}\right) - \mathbb{P}\left( \Phi > z_{\alpha/(2p_1p_2)} + \frac{\sqrt{n}\gamma}{\sigma[\Sigma^{-1}_{k,k}]^{1/2}}\right) \\
& \leq  \mathbb{P}\left( \Phi >  \frac{\sqrt{n}\gamma}{\sigma[\Sigma^{-1}_{k,k}]^{1/2}}-z_{\alpha/(2p_1p_2)} \right),
\end{split}
\end{equation}
where we use $\Phi$ to denote the random variable following a standard Gaussian distribution.

Note that the following inequality holds for standard Normal percentiles:
\begin{equation*}
2e^{-t^2}\leq \mathbb{P}(|\Phi|>t) \leq e^{-t^2/2},
\end{equation*}
and by taking the inverse function, the following inequality holds:
\begin{equation*}
\sqrt{-\log \frac{y}{2}} \leq z_{y/2} \leq \sqrt{-2\log y}.
\end{equation*}
Letting $y=\frac{\alpha}{p_1p_2}$, it follows that:
\begin{equation*}
\left( -\log\frac{\alpha}{2p_1p_2}\right)^{1/2} \leq z_{\alpha/(2p_1p_2)}\leq \left( -2\log \frac{\alpha}{p_1p_2} \right)^{1/2},
\end{equation*}
hence
\begin{equation*}
\mathbb{P}\left(  \Phi >\frac{\sqrt{n}\gamma}{\sigma[\Sigma^{-1}_{k,k}]^{1/2}}-z_{\alpha/(2p_1p_2)}  \right)\leq \mathbb{P}\left(  \Phi >\frac{\sqrt{n}\gamma}{\sigma[\Sigma^{-1}_{k,k}]^{1/2}}-\sqrt{-2\log \frac{\alpha}{p_1p_2}}\right)
\end{equation*}
Now given 
\begin{equation*}
\frac{\log(p_1p_2) }{n} \rightarrow 0, 
\end{equation*}
the following expression follows:
\begin{equation*}
\frac{\sqrt{2\log\left(\frac{p_1p_2}{\alpha}\right)}}{\sqrt{n}/\sigma[\Sigma^{-1}_{k,k}]^{1/2}}\rightarrow 0,
\end{equation*}
indicating that asymptotically, $\left(\frac{\sqrt{n}\gamma}{\sigma[\Sigma^{-1}_{k,k}]^{1/2}}-\sqrt{-2\log \frac{\alpha}{p_1p_2}}\right) \sim \sqrt{n}$. On the other hand, using the fact that $\mathbb{P}(\Phi >t)\leq e^{-t^2/2}$, the last expression in (\ref{type2:ind}) can be  bounded by:
\begin{equation*}
\begin{split}
\mathbb{P}\left( \Phi >  \frac{\sqrt{n}\gamma}{\sigma[\Sigma^{-1}_{k,k}]^{1/2}}-z_{\alpha/(2p_1p_2)} \right) & \leq \exp\left[ -\frac{1}{2} \left(\frac{\sqrt{n}\gamma}{\sigma[\Sigma^{-1}_{k,k}]^{1/2}}-\sqrt{-2\log \frac{\alpha}{p_1p_2}}\right)^2 \right] \\
& \sim e^{-\check{c}n}, \quad \text{for some constant }\check{c}>0.
\end{split}
\end{equation*}
Now with $\log(p_1p_2)/n = o(1)$ and the given sparsity level, that is, $s^* = o(\sqrt{n}/\log p_1)$, it follows that:
\begin{equation*}
\frac{e^{-\check{c}n}}{1/(s^*p_2)} = o(1),
\end{equation*}
i.e., $$(s^*p_2)\cdot \mathbb{P}\left(\widehat{T}_{j,k}=0\right)\rightarrow 0,\qquad \forall j=1,\cdots,p_2; k\in S_j^*.$$
Combining with (\ref{type2}), we have:
\begin{equation*}
\mathbb{P}\left(\text{family-wise type II error}\right) \rightarrow 0, \quad \Leftrightarrow \quad \mathbb{P}\left(\text{family-wise power}\right) \rightarrow 1.
\end{equation*}
This is equivalent to establishing that, given $\log(p_1p_2)/n\rightarrow 0$, the screening step recovers the true support sets $S_j^*$ for all $j=1,2,\cdots,p_2$ with high probability, while keeping the family-wise type I error rate under control. 
\end{proof}

\begin{remark}
The specified level for sparsity is necessary for the de-biased Lasso procedure in \citet{javanmard2014confidence} to produce unbiased estimates for the regression coefficients. In terms of support recovery for the screening step, with $\log(p_1p_2)/n = o(1)$, we only require $s^*=o(p_1)$, which is much weaker and easily satisfied. 
\end{remark}

The following corollary connects the screening step with the alternating search step, under the discussed asymptotic regime : 
\begin{corollary}\label{cor:connection1}
Consider the model set-up given in Section~\ref{sec:set-up}. Let $s^*$ denote the maximum sparsity for all $B^*_j, j=2,\cdots,p_2$, and $d$ denote the maximum degree of $\Theta^*_\epsilon$. Also, let $s^{**}$ denote the sparsity for $\beta^*$ and $s^*_\epsilon$ denote the sparsity for $\Theta^*_\epsilon$. Assume there exist positive constants $c_{s^*},c_{s^{**}}, c_d, c_{\bar{s}},c_{p_2}$ satisfying:
\begin{equation*}
0<c_{s^*}+c_{\bar{s}}<1/2; ~~ 0<c_{s^{**}}+ c_{\bar{s}} <1; ~~ 0<2c_{d}+c_{\bar{s}} <1;~~ 0< \max\{c_{s^*_\epsilon},c_{p_2}\}+ c_{\bar{s}}<1
\end{equation*}
such that 
\begin{equation*}
s^*=O(n^{c_s});~~ s^{**}=O(n^{c_{s^{**}}}); ~~ s_\epsilon^*=O(n^{c_{s^*_\epsilon}});~~  d=O(n^{c_d});~~ \bar{s} = O(e^{n^{c_{p_1}}}); ~~ p_2 = O(n^{c_{p_2}}).
\end{equation*}
As $n\rightarrow\infty$, 
\begin{equation*}
\mathbb{P}\left( \{\text{The screening step correctly recovers the true support set for all $B_j,j=1,\cdots,p$} \}\right) \rightarrow 1, 
\end{equation*}
and for all iterations $k$:
\begin{equation*}
\max\limits_{k\geq 1} \left\|(\widehat{\beta}_{\mathrm{R}},\widehat{\Theta}_\epsilon^{(k)}) - (\beta^*_{\mathrm{R}},\Theta^*_\epsilon)\right\| \stackrel{p}{\rightarrow} 0.
\end{equation*}
\end{corollary}
The proof of this corollary follows along the same lines as Theorem~\ref{thm:beta-theta-bound}, and we leave the details to the reader.

\subsection{Estimation Error and Identifiability}\label{sec:identifiability}

In this subsection, we discuss in detail the conditions needed for the parameters in our multi-layered network to be identifiable (estimable). 
We focus the presentation for ease of exposition on a three-layer network and then discuss the general $M$-layer case.

Consider a $3$-layer graphical model. Let $\widetilde{X} = [(X^1)', (X^2)']'$ be the $(p_1+p_2)$ dimensional random variable, which represents the 
``super"-layer on which we regress $X^3$ to estimate $B^{13}$, $B^{23}$ and $\Sigma^3$. As shown in Theorem~\ref{thm:ErrorBound_beta}, the estimation error for $\widehat{\beta}$ takes the following form:
\begin{equation*}
\|\widehat{\beta}-\beta^*\|_1 \leq 64s^{**}\lambda_n/\varphi 
\end{equation*}
where $\varphi$ is the curvature parameter for RE condition that scales with $\Lambda_{\min}(\Sigma_{\widetilde{X}})$ (see Proposition~\ref{prop:REcondition}). Therefore, the error of estimating these regression parameters is higher when $\Lambda_{\min}(\Sigma_{\tilde{X}})$ is smaller. In this section, we derive a lower bound on this quantity to demonstrate how the estimation error depends on the underlying structure of the graph.

For the undirected subgraph within a layer $k$, we denote its maximum node capacity by $\mathbf{v}(\Theta^k):= \max_{1 \le i \le p_k} \sum_{j=1}^{p_k} |\Theta_{ij}|$. For the directed bipartite subgraph consisting of Layer $s \rightarrow t$ edges ($s<t$), we similarly define the maximum incoming and outgoing node capacities by $\mathbf{v}_{in}(B^{st}):= \max_{1 \le j \le p_t} \sum_{i=1}^{p_s}|B^{st}_{ij}|$ and $\mathbf{v}_{out}(B^{st}):= \max_{1 \le i \le p_s} \sum_{j=1}^{p_t} |B^{st}_{ij}|$. The following proposition establishes the lower bound in terms of these node capacities

\begin{proposition}\label{prop:RE-bound}
\begin{equation*}
\Lambda_{\min}(\Sigma_{\widetilde{X}}) \ge \mathbf{v}(\Theta^1)^{-1} \mathbf{v}(\Theta^2)^{-1} 
\left[1+\left(\mathbf{v}_{in}(B^{12})+\mathbf{v}_{out}(B^{12})\right)/2 \right]^{-2}
\end{equation*}
\end{proposition}
\begin{proof}
From the structural equations of a multi-layered graph introduced in Section~\ref{sec:set-up}, and setting $\epsilon^1:= X^1$, we can write 
\begin{equation}\label{eqn:invert-SEM}
\left[\begin{array}{c}\epsilon^1 \\ \epsilon^2 \end{array} \right] = 
\left[\begin{array}{cc}I & 0 \\ -(B^{12})' & I \end{array}\right] 
\left[\begin{array}{c} X^1 \\ X^2 \end{array} \right]
\end{equation}
Define $P = [I, 0; -(B^{12})', 0]$. Then,  $P \widetilde{X}$ is a centered Gaussian random vector with a block diagonal variance-covariance matrix  $diag(\Sigma^1, \Sigma^2)$. Hence, the concentration matrix of $\widetilde{X}$ takes the form
\begin{equation*}
\Theta_{\tilde{X}} = \Sigma^{-1}_{\tilde{X}} = \left[ \begin{array}{cc}I & -B^{12} \\ 0 & I\end{array} \right] \left[ \begin{array}{cc} \Theta^1 & 0 \\ 0 & \Theta^2 \end{array} \right] \left[\begin{array}{cc}I & 0 \\ -(B^{12})' & 0  \end{array} \right]
\end{equation*}
This leads to an upper bound  
\begin{equation*}
\| \Theta_{\widetilde{X}} \| \le \| \Theta^1 \| \| \Theta^2 \| \|P\|^2
\end{equation*}
The result then follows by using the matrix norm inequality $\|A\| \le \sqrt{\|A\|_1 \|A\|_{\infty}}$ \citep{golub2012matrix}, where $\|A\|_1$ and $\|A\|_{\infty}$ denote the maximum absolute row and column sums of $A$, and the fact that $\Lambda_{\min}(\Sigma_{\tilde{X}}) = \| \Theta_{\tilde{X}} \|^{-1}$. 
\end{proof}

The three components in the lower bound demonstrate how the structure of Layers 1 and 2 impact the accurate estimation of directed edges to Layer 3. Essentially, the bound suggests that accurate estimation is possible when the total effect (incoming and outgoing edges) at every node of each of the three subgraphs is not very large. 

This is inherently related to the identifiability of the multi-layered graphical models and our ability to distinguish between the parents from different layers. For instance, if a node in Layer $2$ has high partial correlation with nodes of Layer $1$, i.e., a node in Layer 2 has parents from many nodes in Layer 1 and yields a large $\mathbf{v}_{in}(B^{12})$; or similarly, a node in Layer $1$ is the parent of many nodes in Layer $2$, yielding a large $\mathbf{v}_{out}(B^{12})$. In either case, we end up with some large lower bound for $\Lambda_{\min}(\Sigma_{\widetilde{X}})$ and it can be hard to distinguish Layer $1 \rightarrow 3$ edges from Layer $2 \rightarrow 3$ edges.

For a general $M$-layer network, the argument in the proof of Proposition \ref{prop:RE-bound} can be genaralized in a straightforward manner, with a modified $P$ of the form 
\begin{equation*}
P = \left[\begin{array}{cccc}
I & 0 & \ldots & 0 \\ 
-(B^{12})' & I & \ldots & 0 \\
\vdots & \vdots & \vdots & 0 \\
-(B^{1, M-1})' & -(B^{2, M-1})' & \ldots & I
 \end{array} \right]
\end{equation*}
and combining node capacities for different layers. The conclusion is qualitatively similar, i.e., the estimation error of a $M$-layer graphical model is small as long as the maximum node capacities of different inter-layer and intra-layer subgraphs are not too large.

\section{Performance Evaluation and Implementation Issues} \label{sec: Implementation}

In this section, we present selected simulation results for our proposed method, in two-layer and three-layer network settings. Further, we introduce some acceleration techniques that can speed up the algorithm and reduce computing time. 

\subsection{\normalsize Simulation Results}
For the 2-layer network, as mentioned in Section \ref{sec:set-up}, since the main target of our proposed algorithm is to provide estimates for $B^*$ and $\Theta_\epsilon^*$ (since $\Theta_X$ can be estimated separately), we only present evaluation results for $B^*$ and $\Theta_\epsilon^*$ estimates. Similarly, for the three-layer network, we only present evaluation results involving Layer 3, using the notation in Section \ref{sec:identifiability}, that is, $B^*_{XZ},B^*_{YZ}$ and $\Theta^*_{\epsilon,Z}$ estimates, which is sufficient to show how our proposed algorithm works in the presence of a 
``super" - layer, taking  advantange of the separability of the log-likelihood. 

\textbf{\textit{2-layered Network.}} To compare the proposed method with the most recent methodology that also provides estimates for the regression parameters and the preccision matrix (CAPME, \cite{cai2012covariate}), we use the exact same model settings that have been used in that paper. Specifically, we consider the following two models: 
\begin{itemize}
\item Model A: Each entry in $B^*$ is nonzero with probability $5/p_1$, and off-diagonal entries for $\Theta_\epsilon^*$ are nonzero with probability $5/p_2$. 
\item Model B: Each entry in $B^*$ is nonzero with probability $30/p_1$, and off-diagonal entries for $\Theta_\epsilon^*$ are nonzero with probability $5/p_2$. 
\end{itemize}
As in \citet{cai2012covariate}, for both models, nonzero entries of $B^*$ and $\Theta_\epsilon^*$ are generated from $\mathsf{Unif}\left[(-1,-0.5)\cup(0.5,1)\right]$, and diagonals of $\Theta_\epsilon^*$ are set identical such that the condition number of $\Theta_\epsilon^*$ is $p_2$. 
\begin{table}[H]
\centering
\caption{Model Dimensions for Model A and B}
\begin{tabular}{cc}
\hline
    	& $(p_1,p_2,n)$ \\ \hline
Model A	 & $p_1=30, p_2=60, n=100$ \\
																		& $p_1=60, p_2=30, n=100$ \\
																		& $p_1=200, p_2=200, n=150$ \\														& $p_1=300, p_2=300, n=150$ \\ 
\hline
Model B& $p_1=200,p_2=200,n=100$ \\
																		& $p_1= 200, p_2=200, n=200$ \\
\hline
\end{tabular}
\end{table}

To evaluate the selection performance of the algorithm, we use sensitivity (SEN), specificity (SPE) and Mathews Correlation Coefficient (MCC) as criteria:
\begin{equation*}
\textrm{SEN} = \frac{\textrm{TN}}{\textrm{TN}+\textrm{FP}},\quad \textrm{SPE} = \frac{\textrm{TP}}{\textrm{TP}+\textrm{FN}}, \quad \textrm{MCC} = \frac{\textrm{TP}\times\textrm{TN}-\textrm{FP}\times \textrm{FN}}{\sqrt{(\textrm{TP}+\textrm{FP})(\textrm{TP}+\textrm{FN})(\textrm{TN}+\textrm{FP})(\textrm{TN}+\textrm{FN})}}.
\end{equation*}
Further, to evaluate the accuracy of the magnitude of the estimates, we use the relative error in Frobenius norm:
\begin{equation*}
\textrm{rel-Fnorm} = \frac{\|\widetilde{B}-B^*\|_\textrm{F}}{\|B^*\|_\textrm{F}}\quad \text{or}\quad \frac{\|\widetilde{\Theta}_\epsilon-\Theta_\epsilon^*\|_\textrm{F}}{\|\Theta_\epsilon^*\|_\textrm{F}}.
\end{equation*}
Tables~\ref{tb:B-eval} and \ref{tb:Theta-eval} show the results for both the regression matrix and the precision matrix. For the precision matrix estimation, we compare our result with those available in \citet{cai2012covariate}, denoted as CAPME.
\begin{table}[h]
\setlength\extrarowheight{2pt}
\centering
\caption{Simulation results for regression matrix over 50 replications}\label{tb:B-eval}
\begin{tabular}{cccccc}
\specialrule{.1em}{0.1em}{0em} 
		& $(p_1,p_2,n)$ &  SEN & SPE & MCC & rel-Fnorm \\  \hline
Model A	 & (30,60,100) &  0.96(0.018) & 0.99(0.004) & 0.93(0.014) & 0.22(0.029) \\
		& (60,30,100) & 0.99(0.009) & 0.99(0.003) & 0.93(0.017) &   0.18(0.021)\\
		& (200,200,150) & 0.99(0.001) & 0.99(0.001) & 0.88(0.009) & 0.18(.007)\\
		& (300,300,150) & 1.00(0.001) & 0.99(0.001) & 0.84(0.010) &  0.21(0.007)\\ 
\hline 
Model B & (200,200,200) & 0.970(0.004) & 0.982(0.001) & 0.927(0.002) & 0.194 (0.009)\\ 
        & (200,200,100) & 0.32(0.010) & 0.99(0.001) & 0.49(0.009) & 0.85(0.006)\\
\specialrule{.1em}{0em}{0.1em}
\end{tabular}\medskip

\caption{Simulation results for precision matrix over 50 replications}\label{tb:Theta-eval}
\begin{tabular}{ccccccc}
\specialrule{.1em}{0.1em}{0em} 
		& $(p_1,p_2,n)$ &  & SEN & SPE & MCC & rel-Fnorm \\  \hline
Model A	 & (30,60,100) &  & 0.77(0.031) & 0.92(0.007) & 0.56(0.030) & 0.51(0.017)\\
		&			   & CAPME & 0.58(0.03) & 0.89(0.01) & 0.45(0.03) &  \\ 
		& (60,30,100) &  & 0.76(0.041) & 0.89(0.015) & 0.59(0.039) & 0.49(0.014) \\
		& (200,200,150) &  & 0.78(0.019) & 0.97(0.001) & 0.55(0.012) & 0.60(0.007) \\
		& (300,300,150) &  & 0.71(0.017)& 0.98(0.001) & 0.51(0.011) & 0.59(0.005)\\ 
\hline 
Model B & (200,200,200) &  & 0.73(0.023) & 0.94(0.003) & 0.39(0.017) & 0.62(0.011)\\ 
		& 				& CAPME & 0.36(0.02) & 0.97(0.00) & 0.35(0.01) &  \\
        & (200,200,100) & 		& 0.57(0.027) & 0.44(0.007) & 0.04(0.008) & 0.84(0.002)\\
        &				& CAPME & 0.19(0.01) & 0.87(0.00) & 0.04(0.01) & \\
\specialrule{.1em}{0em}{0.1em}
\end{tabular}
\end{table}

As it can be seen from Tables~\ref{tb:B-eval} and \ref{tb:Theta-eval}, the sample size is a key factor that affects the performance. Our proposed algorithm performs extremely well in its selection properties on $B$ and strikes a good balance between sensitivity and specificity in estimating $\Theta_\epsilon$\footnote{We suggest using $\alpha=0.1$ as the FWER thresholding level. For tuning parameter selection, we suggest doing a grid search for $(\lambda_n,\rho_n)$ on $[0,0.5\sqrt{\log p_1/n}]\times[0,0.5\sqrt{\log p_2/n}]$ with BIC.}. For most settings, it provides substantial 
improvements over the CAPME estimator.

\medskip
\textbf{\textit{3-layer Network.}} For a 3-layer network, we consider the following data generation mechanism: for all three models A, B and C, each entry in $B_{XY}$ is nonzero with probability $5/p_1$, each entry in $B_{XZ}$ and $B_{YZ}$ is nonzero with probability $5/(p_1+p_2)$, and off-diagonal entries in $\Theta_{\epsilon,Z}$ are nonzero with probability $5/p_3$. Similar to the 2-layered set-up, the nonzero entries in $\Theta_{\epsilon,Z}$ are generated from $\mathsf{Unif}[(-1,-0.5)\cup(0.5,1)]$ with its diagnals set identical such that its condition number is $p_3$. For the regression matrices in the three models, nonzeros in $B_{XY}$ are generated from $\mathsf{Unif}[(-1,-0.5)\cup(0.5,1)]$, and nonzeros in $B_{XZ}$ and $B_{YZ}$ are generated from $\left\{\mathsf{Unif}[(-1,-0.5)\cup(0.5,1)]* \text{Signal.Strength}\right\}$, where the signal strength in the three models are given by 1, 1.5 and 2, respectively. More specifically, for Model A, B and C, nonzeros in $B_{XZ}$ or $B_{YZ}$ are generated from $\mathsf{Unif}[(-1,-0.5)\cup(0.5,1)]$, $\mathsf{Unif}[(-1.5,-0.75)\cup(0.75,1.5)]$ and $\mathsf{Unif}[(-2,-1)\cup(1,2)]$, respectively. 
\begin{table}[H]
\setlength\extrarowheight{2pt}
\centering
\caption{Model Dimensions and Signal Strength for Model A, B and C}
\begin{tabular}{ccc}
\hline
 & Layer 3 Signal.Strength & $(p_1,p_2,p_3,n)$ \\  \hline
 Model A & $1$ & (50,50,50,200) \\
 Model B &  $1.5$ & (50,50,50,200) \\
 Model C & $2$ & (50,50,50,200) \\
		 &  	& (20,80,50,200) \\
		 &  	& (80,20,50,200) \\
		 &  	& (100,100,100,200) \\
\hline
\end{tabular}
\end{table}\medskip

As mentioned in the beginning of this subsection, we only evaluate the algorithm's performance on $B_{XZ},B_{YZ}$ and $\Theta_{\epsilon,Z}$. 

\begin{table}[h]
\setlength\extrarowheight{2pt}
\centering
\caption{Simulation results for regression matrix $B_{XZ}$ over 50 replications}\label{tb:Bxz-eval}
\begin{tabular}{cccccc}
\specialrule{.1em}{0.1em}{0em} 
		& $(p_1,p_2,p_3,n)$ &  SEN & SPE & MCC & rel-Fnorm \\  \hline
Model A	 & (50,50,50,200) & 0.51(0.065) & 0.99(0.001)& 0.69(0.049) & 0.68(0.050) \\
Model B & (50,50,50,200)  & 0.85(0.043) & 0.99(0.001) & 0.898(0.025) & 0.36(0.056) \\
Model C & (50,50,50,200) &  0.97(0.018) & 0.99(0.002) & 0.96(0.016) & 0.16(0.040) \\
		& (20,80,50,200) & 0.55(0.078) & 0.99(0.001) & 0.72(0.059) & 0.63(0.066) \\
		& (80,20,50,200) & 0.99(0.006) & 0.99(0.002) & 0.94(0.017) & 0.076(0.032) \\ 
		& (100,100,100,200)& 1.00(0.001) & 0.99(0.001) & 0.87(0.016) & 0.07(0.007)\\
\specialrule{.1em}{0em}{0.1em}
\end{tabular}\medskip

\caption{Simulation results for regression matrix $B_{YZ}$ over 50 replications}\label{tb:Byz-eval}
\begin{tabular}{cccccc}
\specialrule{.1em}{0.1em}{0em} 
		& $(p_1,p_2,p_3,n)$ &  SEN & SPE & MCC & rel-Fnorm \\  \hline
Model A	 & (50,50,50,200) &  0.53(0.051) & 1.00(0.000) & 0.72(0.036) & 0.65(0.041) \\
Model B & (50,50,50,200)  & 0.90(0.033) & 1.00(0.000) & 0.95(0.019) & 0.25(0.049) \\
Model C & (50,50,50,200) &  0.98(0.013) & 1.00(0.000) & 0.99(0.007) & 0.12(0.042) \\
		& (20,80,50,200) & 0.95(0.013) & 1.00(0.000) & 0.98(0.007) & 0.19(0.030) \\
		& (80,20,50,200) & 0.96(0.027) & 0.99(0.001) & 0.97(0.022) & 0.14(0.063) \\
		& (100,100,100,200) & 1.00(0.000) & 1.00(0.000) & 0.99(0.002) & 0.025(0.002)  \\
\specialrule{.1em}{0em}{0.1em}
\end{tabular}\medskip

\caption{Simulation results for regression matrix $\Theta_{\epsilon,Z}$ over 50 replications}\label{tb:ThetaZ-eval}
\begin{tabular}{cccccc}
\specialrule{.1em}{0.1em}{0em} 
		& $(p_1,p_2,p_3,n)$ &  SEN & SPE & MCC & rel-Fnorm \\  \hline
Model A	 & (50,50,50,200) & 0.69(0.044) & 0.638(0.032) & 0.20(0.036) & 0.82(0.017) \\
Model B & (50,50,50,200)  & 0.77(0.050) & 0.82(0.036) & 0.42(0.071) & 0.68(0.040) \\
Model C & (50,50,50,200) &  0.88(0.041) & 0.91(0.019) & 0.63(0.059) & 0.56(0.034) \\
		& (20,80,50,200) & 0.72(0.041) & 0.80(0.028) & 0.36(0.050) & 0.72(0.021) \\
		& (80,20,50,200) &  0.90(0.028) & 0.92(0.011) & 0.68(0.039) & 0.58(0.018) \\ 
		& (100,100,100,200) & 0.96(0.014) & 0.96(0.003) & 0.68(0.016) & 0.049(0.010)\\
\specialrule{.1em}{0em}{0.1em}
\end{tabular}
\end{table}\medskip
Based on the results shown in Tables~\ref{tb:Bxz-eval}, \ref{tb:Byz-eval} and \ref{tb:ThetaZ-eval}, the signal strength across layers affects the accuracy of the estimation, which is in accordance with what has been discussed regarding identifiability. Overall, the MLE estimator performs satisfactorily across a fairly wide range of settings and in many cases achieving very high values for the MCC criterion.

\subsubsection{Simulation Results for non-Gaussian data}

In many applications, the data may not be exactly Gaussian, but approximately Gaussian. Next, we present selected  simulation results when the data comes from some distribution that deviates from Gaussian. Specifically, we consider two types of deviations based on the following transformations: (i) a truncated empirical cumulative distribution function and (ii) a shrunken empirical cumulative distribution functions as discussed in \citet{huge}. In both simulation settings, we consider Model A with $(p_1,p_2,n)=(30,60,100)$ under the two-layer setting, and the transformation is applied to errors in Layer 2. Table~\ref{tb:npn} shows the simulation results for these two scenarios over 50 replications.
\begin{table}
\caption{Simulation results for $B$ and $\Theta_{\epsilon}$ over 50 replications under npn transformation}\label{tb:npn}
\begin{tabular}{c|c|cccc}
\specialrule{.1em}{0.1em}{0em} 
Setting & Parameter& SEN & SPE & MCC & rel-Fnorm \\  \hline
Model A $(30,60,100)$  & $B$ & 0.96(0.017) & 0.99(0.003) & 0.94(0.012) &   0.20(0.028) \\ 
shrunken & $\Theta_\epsilon$ & 0.76(0.031) & 0.91(0.008) & 0.55(0.030) & 0.51(0.019) \\ \hline
Model A $(30,60,100)$  & $B$ & 0.96(0.021) & 0.98(0.004) & 0.93(0.015) & 0.21(0.034) \\ 
truncation & $\Theta_\epsilon$ & 0.76(0.033) & 0.92(0.008) & 0.56(0.035) & 0.52(0.023) \\
\specialrule{.1em}{0em}{0.1em}
\end{tabular}
\end{table}

Based on the results in Table~\ref{tb:npn}, relatively small deviatiosn from the Gaussian distribution does not affect the performace of the MLE estimates under the examined settings that are comparable to those obtained with Gaussian distributed data.

\subsection{A comparison with the two-step estimator in \citet{cai2012covariate}}

Next,  we present a comparison between the MLE estimator and the two-step estimator of \citet{cai2012covariate}. Specifically, we use the CAPME estimate as an initializer for the MLE procedure and examine its evolution over successive iterations.  We evaluate the value of the objective function at each iteration, and also compare it to the value of the objective function evaluated at our initializer (screening $+$ Lasso/Ridge) and the estimates afterwards. For illustration purposes, we only show the results for a single relaization under Model A with $p_1=30,p_2=60,n=100$, although similar results were obtained in other simulation settings. Figure~\ref{fig:comparison} shows the value of the objective function as a function of the iteration under both initialization procedures, while Table~\ref{tb:comparison} shows how the cardinality of the estimates changes over iterations for both initializers. It can be seen that the iterative MLE algorithm significantly improves the value of the objective function over the CAPME initialization and also that the set of directed and undirected edges stabilizes after a couple iterations.
\begin{figure}[htbp]
\centering
\caption{Comparison between Cai's estimate and our estimate}\label{fig:comparison}
\includegraphics[scale=0.6]{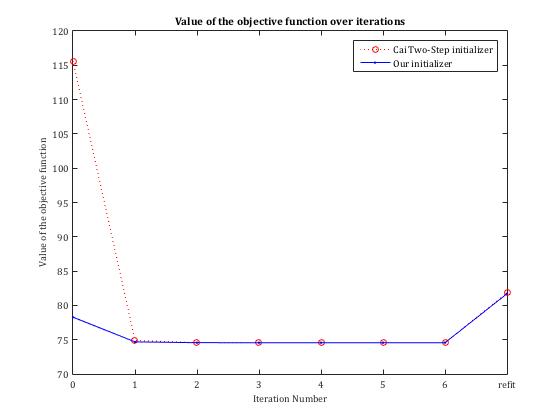}
\end{figure} 
\begin{table}[h]
\centering
\caption{Change in cardinality over iterations for $B$ and $\Theta_\epsilon$} \label{tb:comparison}
\begin{tabular}{cccccccccc}
\specialrule{.1em}{0.1em}{0em} 
	& 	 & 	0 & 1 & 2 & 3 & 4 & 5 & 6 & refit \\ \hline
Our  initializer & $\widehat{B}^{(k)}$ & 275 & 275& 275& 275& 275 & 275 & 275 &  275 \\ 
	& $\widehat{\Theta}_\epsilon^{(k)}$ & 282 & 255&  247 & 247 & 248 & 248  & 248 & 260\\
CAPME initializer & $\widehat{B}^{(k)}$ &  433 & 275& 275& 275& 275 & 275 & 275 &  275 \\
					& $\widehat{\Theta}_\epsilon^{(k)}$ & 979 & 267  & 250 & 249 & 249 & 248 & 248 & 260 \\
\specialrule{.1em}{0.1em}{0em} 
\end{tabular} 
\end{table}

Based on Figure~\ref{fig:comparison} and Table~\ref{tb:comparison}, we notice that Cai et. al's two-step estimator yields larger value of the objective function compared with our initializer that is obtained through screening followed by Lasso. However, over subsequent iterations, both initializers yield the same value in the objective function, which keeps decreasing according to the nature of block-coordinate descent.  

\subsection{\normalsize Implementation issues}

Next, we introduce some acceleration techniques for the MLE algorithm aiming to reduce  computing time, yet maintaining estimation accuracy over iterations.

\medskip
\textbf{\textit{$\bs{(p_2+1)}$-block update.}} In Section \ref{sec:methodology}, we update $B$ and $\Theta_\epsilon$ by (\ref{eqn:updateB}) and (\ref{eqn:updateTheta}), respectively, and within each iteration, the updated $B$ is obtained by an application of cyclic $p_2$-block coordinate descent with respect to each of its columns until convergence. As shown in Section \ref{sec:convergence}, the outer 2-block update guarantees the MLE iterative algorithm to converge to a stationary point. However in practice, we can speed up the algorithm by updating $B$ without waiting for it to reach the minimizer for every iteration other than the first one. More precisely, for the alternating search step, we take the following steps when actually implementing the proposed algorithm with initializer $\widehat{B}^{(0)}$ and $\widehat{\Theta}_\epsilon^{(0)}$: 
\begin{itemize}
\item[--] Iteration 1: update $B$ and $\Theta_\epsilon$ as follows, respectively:
\begin{equation*}
\widehat{B}^{(1)} = \argmin\limits_{B\in \mathcal{B}_1\times \cdots\times \mathcal{B}_{p_2}} \left\{ \frac{1}{n}\sum_{i=1}^{p_2}\sum_{j=1}^{p_2}(\sigma_{\epsilon}^{ij})^{(0)}(Y_i-XB_i)^\top (Y_j - XB_j) + \lambda_n\sum_{j=1}^{p_2}\|B_j\|_1 \right\},
\end{equation*}
and
\begin{equation*}
 \widehat{\Theta}_\epsilon^{(1)} = \argmin\limits_{\Theta_\epsilon\in\mathbb{S}_{++}^{p_2\times p_2}} \left\{
 \log\det\Theta_\epsilon - \text{tr}(\widehat{S}^{(1)}\Theta_\epsilon) + \rho_n\|\Theta_\epsilon\|_{1,\text{off}}
 \right\},
 \end{equation*}
where $\widehat{S}^{(1)}$ is the sample covariance matrix of $\widehat{E}^{(1)}\equiv Y- X\widehat{B}^{(1)}$. 
\item[--] For iteration $k\geq 2$, while not converged: 
\begin{itemize}
\item[$\cdot$] For $j=1,\cdots,p_2$, update $B_j$ once by:
\begin{equation*}
\widehat{B}_j^{(k)} = \argmin\limits_{B_j\in\mathcal{B}_j} \left\{ \frac{(\sigma^{jj}_\epsilon)^{(k-1)}}{n}\|Y_j+r_j^{{(k)}}-XB_j\|_2^2 + \lambda_n\|B_j\|_1 \right\},
\end{equation*}
where 
\begin{equation}\label{eqn:updater}
r_j^{(k)} = \frac{1}{(\sigma^{jj}_\epsilon)^{(k-1)}}\left[ \sum_{i=1}^{j-1} (\sigma^{ij}_\epsilon)^{(k-1)}(Y_i-X\widehat{B}_i^{(k)}) + \sum_{i=j+1}^{p_2} (\sigma^{ij}_\epsilon)^{(k-1)} (Y_i-X\widehat{B}_i^{(k-1)})\right].
\end{equation}
\item[$\cdot$] Update $\Theta_\epsilon$ by:
\begin{equation*}
 \widehat{\Theta}_\epsilon^{(k)} = \argmin\limits_{\Theta_\epsilon\in\mathbb{S}_{++}^{p_2\times p_2}} \left\{
 \log\det\Theta_\epsilon - \text{tr}(\widehat{S}^{(k)}\Theta_\epsilon) + \rho_n\|\Theta_\epsilon\|_{1,\text{off}}
 \right\},
 \end{equation*}
 where $\widehat{S}^{(k)}$ is defined similarly. 
\end{itemize}
\end{itemize}
Intuitively, for the first iteration, we wait for the algorithm to complete the whole cyclic $p_2$ block-coordinate descent step, as the first iteration usually achieves a big improvement in the value of the objective function compared to the initialization values, as depicted in Figure~\ref{fig:comparison}.
However, in subsequent iterations, the changes in the objective function become relatively small, so we do $(p_2+1)$ successive block-updates in every iteration, and start to update $\Theta_\epsilon$ once a full $p_2$ block update in $B$ is completed, instead of waiting for the update in $B$ proceeds cyclically until convergence. In practice, this way of updating $B$ and $\Theta_\epsilon$ leads to faster convergence in terms of total computing time, yet yields the same estimates compared with the exact $2$-block update shown in Section \ref{sec:methodology}.   \\

\medskip
\textbf{\textit{Parallelization.}} A number of steps of the MLE algorithm is parallelizable. In the screening step, when applying the de-biased Lasso procedure \citep{javanmard2014confidence} to obtain the $p$-values, we need to implement $p_2$ separate regressions, which can be distributed to different compute nodes and carried out in parallel. So does the refitting step, in which we refit each column in $B$ in parallel. 

Moreover, according to \citet{bradley2011parallel,richtarik2012parallel,scherrer2012scaling} and a series of similar studies, though the block update in the alternating search step is supposed to be carried out sequentially, we can implement the update in parallel to speed up convergence, yet empirically yield identical estimates. This parallelization can be applied to either the minimization with respect to $B$ within the 2-block update method, or the minimization with respect to each column of $B$ for the $(p_2+1)$-block update method. Either way, $r_j^{(k)}$ in (\ref{eqn:updater}) is substituted by
\begin{equation*}
r_{j,\text{parallel}}^{(k)} = \frac{1}{(\sigma^{jj}_\epsilon)^{(k-1)}}\sum_{i\neq j}^{p_2} (\sigma_\epsilon^{ij})^{(k-1)}(Y_i-X\widehat{B}_i^{(k-1)}),
\end{equation*} 
which is not updated until we have updated $B_j$'s once for all $j=1,\cdots,p_2$ in parallel.\\

\medskip
The table below shows the elapsed time for carrying out our proposed algorithm using 2-block/$(p_2+1)$ -block update with/without parallelization, under the simulation setting where we have $p_1=p_2=200,n=150$. The screening step and refitting step are both carried out in parallel for all four different implementations\footnote{For parallelization, we distribute the computation on 8 cores.}. 
\begin{table}[H]
\centering
\caption{Computing time with different update methods}
\begin{tabular}{c|cccc}
\hline
	& 2-block & $(p_2+1)$-block & 2-block in parallel & $(p_2+1)$-block in parallel \\
\hline
elasped time (sec) & 5074 & 2556 & 848 & 763 \\
\hline 
\end{tabular}
\end{table}

As shown in the table, using $(p_2+1)$-block update and parallelization both can speed up convergence and reduce computing time, which takes only 1/7 of the computing time compared with using 2-block update without parallelization.

\begin{remark}
The total computing time depends largely on the number of bootstrapped samples we choose for the stability selection step. For the above displayed results, we used 50 bootstrapped samples to obtain the weight matrix. Nevertheless, one can select the number of bootstrap samples judiciously and reduce them if performance would  not be seriously impacted.
\end{remark}

\newpage
\bibliographystyle{plainnat}
\bibliography{00-bib.bib}

\newpage
\section{Appendix}
\subsection{\normalsize Proofs for Propositions and Auxillary Lemmas}
To prove Proposition~\ref{prop:REcondition}, we need the following two lemmas. Lemma~\ref{lemma:restateB.1} was originally provided as Lemma~B.1 in \citet{basu2015estimation}, which states that if the sample covariance matrix of $X$ satisfies the RE condition and $\Theta$ is diagonally dominant, then $(X'X/n)\otimes \Theta$ also satisfies the RE condition. Here we omit its proof and only state the main result. Lemma~\ref{lemma:RESX} verifies that with high probability, the sample covariance matrix of the design matrix $X$ satisfies the RE condition. 

\medskip
\begin{lemma}\label{lemma:restateB.1}
If $X'X/n\sim RE(\varphi^*,\phi^*)$, and $\Theta$ is diagonally dominant, that is, $\psi^i:= \sigma^{ii}-\sum_{j\neq i} \sigma^{ij}>0$ for all $i=1,2,\cdots,p_2$, where $\sigma^{ij}$ is the $ij$th entry in $\Theta$, then
\begin{equation*}
\Theta\otimes X'X/n \sim RE\left(\varphi^*\min_i \psi^i,\phi^*\max_i\psi^i\right).
\end{equation*}
\end{lemma}
\medskip 

\begin{lemma}\label{lemma:RESX}
With probability at least $1-2\exp(-c_3n)$, for a zero-mean sub-Gaussian random design matrix $X\in\mathbb{R}^{n\times p_1}$, its sample covariance matrix $\widehat{\Sigma}_X$ satisfies the RE condition with parameter $\varphi^*$ and $\phi^*$, i.e.,
\begin{equation}
\widehat{\Sigma}_X\sim RE(\varphi^*,\phi^*),
\end{equation}
where $\widehat{\Sigma}_X=X'X/n$, $\varphi^*=\Lambda_{\min}(\Sigma^*_X)/2$, $\phi^*=\varphi^* \log p_1/n$.
\end{lemma}
\begin{proof}
To prove this lemma, we first use Lemma 15 in \citet{loh2011high}, which states that if $X\in\mathbb{R}^{n\times p}$ is  zero-mean sub-Gaussian with parameter $(\Sigma,\sigma^2)$, then there exists a universal constant $c>0$ such that 
\begin{equation}
\mathbb{P}\left( \sup_{v\in\mathbb{K}(2s)}\left| \frac{\|Xv\|_2^2}{n} - \mathbb{E}\left[ \frac{\|Xv\|_2^2}{n}\right] \right| \geq t\right) \leq 2\exp\left( -cn \min(\frac{t^2}{\sigma^4},\frac{t}{\sigma^2})+ 2s\log p\right),
\end{equation}
where $\mathbb{K}(2s)$ is a set of $2s$ sparse vectors, defined as:
\begin{equation*}
\mathbb{K}(2s) : = \{v\in\mathbb{R}^p:\|v\|\leq 1,\|v\|_0\leq 2s\}.
\end{equation*}
By taking $t=\frac{\Lambda_{\min}(\Sigma^*_X)}{54}$, with probability at least $1-2\exp\left(-c'n + 2s\log p_1\right)$ for some $c'>0$, the following bound holds:
\begin{equation}\label{dev:Sx}
|v' (\widehat{\Sigma}_X-\Sigma^*_X) v | \leq \frac{\Lambda_{\min}(\Sigma^*_X)}{54}, \quad \forall v\in \mathbb{K}(2s).
\end{equation}
Then applying supplementary Lemma~13 in \citet{loh2011high}, for an estimator $\widehat{\Sigma}_X$ of $\Sigma^*_X$ satisfying the deviation condition in (\ref{dev:Sx}), the following RE condition holds:
\begin{equation*}
v'S_xv \geq \frac{\Lambda_{\min}(\Sigma^*_X)}{2}\|v\|_2^2 - \frac{\Lambda_{\min}(\Sigma^*_X)}{2s}\|v\|_1^2.
\end{equation*}
Finally, set $s=c''n/4\log p_1$, then with probability at least $1-2\exp(-c_3n)~(c_3>0)$, $\widehat{\Sigma}_X\sim RE(\varphi^*,\phi^*)$ with $\varphi^*=\Lambda_{\min}(\Sigma^*_X)/2$, $\phi^*= \varphi^*\log p_1/n$. 
\end{proof}

\medskip

With the above two lemmas, we are ready to prove Proposition~\ref{prop:REcondition}. 

\begin{proof}[Proof of Proposition~\ref{prop:REcondition}]
We first show that if $\Theta^*_\epsilon$ is diagonally dominant, then $\widehat{\Theta}_\epsilon$ is also diagonally dominant provided that the error of $\widehat{\Theta}_\epsilon$ is of the given order and $n$ is sufficiently large. Define 
\begin{equation*}
\widehat{\psi}^i = \widehat{\sigma}_\epsilon^{ii} - \sum\limits_{j\neq i}\widehat{\sigma}_\epsilon^{ij},
\end{equation*}
where $\widehat{\sigma}_\epsilon^{ij}$ is the $ij$th entry of $\widehat{\Theta}_\epsilon$, then $\widehat{\psi}^i$ is the gap between the diagonal entry and the off-diagonal entries of row $i$ in matrix $\widehat{\Theta}_\epsilon$. We can decompose $\widehat{\psi}^i$ into the following:
\begin{equation*}
\widehat{\psi}^i = \left[\sigma^{ii}_\epsilon - \sum\limits_{j\neq i}\sigma^{ij}_\epsilon\right] + \left[ (\widehat{\sigma}^{ii}_\epsilon - \sigma^{ii}_\epsilon) + \sum\limits_{j\neq i} (\sigma^{ij}_\epsilon - \widehat{\sigma}^{ij}_\epsilon)\right].
\end{equation*}
Recall that we define $\psi_i$ as $\psi^i=\sigma_\epsilon^{ii}-\sum_{j\neq i}^{p_2}\sigma_\epsilon^{ij}$. Hence
\begin{equation}\label{eqn:boundpsi}
\begin{split}
\min\widehat{\psi}^i& \geq \min_i \psi^i - \vertiii{\widehat{\Theta}_\epsilon-\Theta_\epsilon^*}_\infty \geq \min_i (\sigma_\epsilon^{ii} - \sum_{j\neq i}\sigma^{ij}_\epsilon) - d\nu_\Theta = \min\psi^i - d\nu_\Theta,\\
\max\widehat{\psi}^i& \leq \max_i \psi^i + \vertiii{\widehat{\Theta}_\epsilon-\Theta_\epsilon^*}_\infty \leq \max_i (\sigma_\epsilon^{ii} - \sum_{j\neq i}\sigma^{ij}_\epsilon) + d\nu_\Theta = \max\psi^i + d\nu_\Theta.
\end{split}
\end{equation}
Now given $\nu_\Theta = \eta_\Theta\frac{\log p_2}{n} = O(\sqrt{\log p_2/n})$, with $n\succsim d^2\log p_2$, $d\nu_\Theta=o(1)$, and it follows that
\begin{equation*}
\min\limits_i \psi^i - d\nu_\Theta \geq 0.
\end{equation*}
Now by Lemma~\ref{lemma:RESX}, $X'X/n\sim RE(\varphi^*,\phi^*)$ with high probability. Combine with Lemma~\ref{lemma:restateB.1} and inequality (\ref{eqn:boundpsi}), with probability at least $1-2\exp(-c_3n)$ for some  $c_3>0$,  $\widehat{\Gamma}$ satisfies the following RE condition:
\begin{equation}
\widehat{\Gamma} = \widehat{\Theta}_\epsilon\otimes (X'X/n)\sim RE\left(\varphi^* (\min_i\psi^i-d\nu_\Theta), \phi^* \max\limits_i(\psi^i+d\nu_\Theta)\right),
\end{equation}
where $\varphi^*=\Lambda_{\min}(\Sigma^*_X)/2$, $\phi^*=\varphi^*\log p_1/n$.
\end{proof}

\medskip

To prove Proposition~\ref{prop:deviation}, we first prove Lemma~\ref{lemma:deviation_aux}.

\begin{lemma}\label{lemma:deviation_aux}
Let $X\in\mathbb{R}^{n\times p}$ be a zero-mean sub-Gaussian matrix with parameter $(\Sigma_X,\sigma_X^2)$ and $E\in\mathbb{R}^{n\times p_2}$ be a zero-mean sub-Gaussian matrix with parameters $(\Sigma_\epsilon,\sigma^2_\epsilon)$. Moreover, $X$ and $E$ are independent. Let $\Theta_\epsilon:=\Sigma_\epsilon^{-1}$, then if $n\succsim \log (p_1p_2)$, the following two expressions hold with probability at least $1-6c_1\exp[-(c_2^2-1)\log (p_1p_2)]$ for some $c_1>0,c_2>1$, respectively:
\begin{equation}\label{deviation_aux-1}
\frac{1}{n}\left\|X'E\right\|_\infty \leq c_2\left[\Lambda_{\max}(\Sigma_X)\Lambda_{\max}(\Sigma_\epsilon)\right]^{1/2} \sqrt{\frac{\log (p_1p_2) }{n}}.
\end{equation}
and
\begin{equation}\label{deviation_aux-2}
\frac{1}{n}\left\| X'E\Theta_\epsilon\right\|_\infty \leq c_2\left[\frac{\Lambda_{\max}(\Sigma_X)}{\Lambda_{\min}(\Sigma_\epsilon)}\right]^{1/2} \sqrt{\frac{\log (p_1p_2) }{n}}.
\end{equation}
\end{lemma}

\begin{proof}
The proof of this lemma uses Lemma 14 in \citet{loh2011high}, in which they show that if $X\in\mathbb{R}^{n\times p_1}$ is a zero-mean sub-Gaussian matrix with parameters $(\Sigma_x,\sigma_x^2)$ and $Y\in\mathbb{R}^{n\times p_2}$ is a zero-mean sub-Gaussian matrix with parameters $(\Sigma_y,\sigma_y^2)$, then if $n\succsim \log(p_1p_2)$, 
\begin{equation*}
\mathbb{P}\left(\left\| \frac{Y'X}{n} - cov(y_i,x_i)\right\|_\infty \geq t \right) \leq 6p_1p_2\exp\left(-cn\min\left\{\frac{t^2}{(\sigma_x\sigma_y)^2},\frac{t}{\sigma_x\sigma_y} \right\} \right)
\end{equation*}
where $X_i$ and $Y_i$ are the $i$th row of $X$ and $Y$, respectively.

Here, we replace $Y$ by $E$, and since $E$ and $X$ are independent, $cov(X_i,E_i)=0$. Let $t=c_2\sigma_X\sigma_\epsilon\sqrt{\log(p_1p_2)/n}$, $c_2>1$ we get:
\begin{equation*}
\mathbb{P}\left(\left\| \frac{X'E}{n}\right\|_\infty\geq c_2\sigma_X\sigma_\epsilon\sqrt{\frac{\log(p_1p_2)}{n}} \right)\leq 6c_1(p_1p_2)^{1-c_2^2} = 6c_1\exp\left[-(c_2^2-2)\log(p_1p_2)\right]
\end{equation*}
Note that the sub-Gaussian parameter satisfies $\sigma^2_X\leq \max_i(\Sigma_{X,ii})\leq \Lambda_{\max}(\Sigma_X)$. This directly gives the bound in (\ref{deviation_aux-1}).

To obtain the bound in (\ref{deviation_aux-2}), we note that if $E$ is sub-Gaussian with parameters $(\Sigma_\epsilon,\sigma_\epsilon^2)$, then $E\Theta$ is sub-Gaussian with parameter $(\Theta,\theta_\epsilon^2)$, where
\begin{equation*}
\theta_\epsilon^2 \leq \max_i(\Theta_{\epsilon,ii}) \leq \Lambda_{\max}(\Theta_\epsilon) = \frac{1}{\Lambda_{\min}(\Sigma_\epsilon)}.
\end{equation*}
Then we replace $Y$ by $E\Theta$ and yield the bound in (\ref{deviation_aux-2}).
\end{proof}

As a remark, here we note that the event in (\ref{deviation_aux-1}) and (\ref{deviation_aux-2}) may not be independent. However, the two events hold simultaneously with probability at least $1-2c_2\exp[-c_2\log(p_1p_2)]$, with this crude bound for probability hold for sure. 

\medskip
Now we are ready to prove Proposition~\ref{prop:deviation}. 

\begin{proof}[Proof of Proposition~\ref{prop:deviation}] 
First we note that 
\begin{equation*}
X'E\widehat{\Theta}_\epsilon= X'E\Theta_\epsilon + X'E (\widehat{\Theta}_\epsilon -\Theta^*_\epsilon),
\end{equation*}
which directly gives the following inequality:
\begin{equation}\label{decomp}
\begin{split}
\|\widehat{\gamma}-\widehat{\Gamma}\beta^*\|_\infty=\frac{1}{n}\left\|X'E\widehat{\Theta}_\epsilon\right\|_\infty  \leq \frac{1}{n} \left\|X'E\Theta_\epsilon^*\right\|_\infty + \frac{1}{n}\left\|X'E (\widehat{\Theta}_\epsilon-\Theta_\epsilon^*)\right\|_\infty.
\end{split}
\end{equation}
Now we would like to bound the two terms separately. 

The first term can be bounded by (\ref{deviation_aux-2}) in Lemma~\ref{lemma:deviation_aux}, that is:
\begin{equation*}
\frac{1}{n}\left\| X'E\Theta_\epsilon^*\right\|_\infty \leq c_2\left[\frac{\Lambda_{\max}(\Sigma_X)}{\Lambda_{\min}(\Sigma^*_\epsilon)}\right]^{1/2} \sqrt{\frac{\log (p_1p_2)}{n}}.
\end{equation*}
w.p. at least $1-6c_1\exp[-(c_2^2-1)\log (p_1p_2)]$.

For the second term, first we note that 
\begin{equation}\label{exp: 2nd term}
\begin{split}
\frac{1}{n}\left\|X'E (\widehat{\Theta}_\epsilon-\Theta_\epsilon^*)\right\|_\infty & = \frac{1}{n}\max\limits_{\substack{1\leq i\leq p_1\\ 1\leq j\leq p_2}} \left|e_i'X'E (\widehat{\Theta}_\epsilon-\Theta_\epsilon^*) e_j\right| \\
& \leq \frac{1}{n}\max\limits_{i} \left\|e_i'X'E \right\|_\infty \max\limits_j \left\| (\widehat{\Theta}_\epsilon-\Theta^*_\epsilon) e_j\right\|_1
\end{split}
\end{equation}
where we have $e_i\in\mathbb{R}^{p_1}$ and $e_j\in\mathbb{R}^{p_2}$, and the inequality comes from the fact that $|a'b|\leq \|a\|_\infty \|b\|_1$. Note that 
\begin{equation*}
\max\limits_{i} \left\|e_i'X'E \right\|_\infty = \|X'E\|_\infty
\end{equation*}
since $\|e_i'X'E \|_\infty$ gives the largest element (in absolute value) of the $i$th row of $X'E$, and taking the maximum over all $i$'s gives the largest element of $X'E$ over all entries. And for $\max\limits_j \left\| (\widehat{\Theta}_\epsilon-\Theta_\epsilon^*) e_j\right\|_1$, it holds that 
\begin{equation*}
\max\limits_j \left\| (\widehat{\Theta}_\epsilon-\Theta_\epsilon^*) e_j\right\|_1 = \vertiii{\widehat{\Theta}_\epsilon-\Theta_\epsilon^*}_1 = \vertiii{\widehat{\Theta}_\epsilon-\Theta_\epsilon^*}_\infty,
\end{equation*}
where $\vertiii{A}_1 :=\max_{\|x\|_1=1}\|Ax\|_1$ is the $\ell_1$-operator norm, and the last equality follows from the fact that $\vertiii{A}_1=\vertiii{A'}_\infty$. As a result,  (\ref{exp: 2nd term}) can be re-written as:
\begin{equation}\label{term2simplified}
\frac{1}{n}\left\|X'E (\widehat{\Theta}_\epsilon-\Theta_\epsilon^*)\right\|_\infty\leq \left(\frac{1}{n}\|X'E\|_\infty\right) \left( \vertiii{\widehat{\Theta}_\epsilon-\Theta_\epsilon^*}_\infty\right).
\end{equation}
Now, using (\ref{deviation_aux-1}), w.p. at least $1-6c_1\exp[-(c_2^2-1)\log (p_1p_2)]$, we have 
\begin{equation*}
\frac{1}{n}\left\|X'E\right\|_\infty \leq c_2\left[\Lambda_{\max}(\Sigma_X)\Lambda_{\max}(\Sigma^*_\epsilon)\right]^{1/2} \sqrt{\frac{\log (p_1p_2)}{n}},
\end{equation*}
and since $\|\widehat{\Theta}_\epsilon-\Theta_\epsilon^*\|_\infty\leq \nu_\Theta$, it directly follows that $\vertiii{\widehat{\Theta}_\epsilon-\Theta_\epsilon^*}_\infty\leq d\nu_\Theta$. Therefore, with probability at least $1-6c_1\exp[-(c_2^2-1)\log (p_1p_2)]$, 
\begin{equation}\label{bound:2nd}
\frac{1}{n}\left\|X'E (\widehat{\Theta}_\epsilon-\Theta_\epsilon^*)\right\|_\infty \leq c_2d\nu_\Theta\left[\Lambda_{\max}(\Sigma_X)\Lambda_{\max}(\Sigma^*_\epsilon)\right]^{1/2} \sqrt{\frac{\log (p_1p_2)}{n}}.
\end{equation}
Combine the two terms, we obtain the conclusion in Proposition~\ref{prop:deviation}.
\end{proof}

\medskip

\begin{proof}[Proof of Corollary~\ref{cor:ErrorBoundB}]
Here we examine the probability that events A1-A3 hold in Theorem~\ref{thm:ErrorBound_beta}. First we note that (A1) in Theorem~\ref{thm:ErrorBound_beta} holds deterministically. Now by Proposition~\ref{prop:REcondition}, (A2) is satisfied w.p. at least $1-2\exp(-c_3n)$. By Proposition~\ref{prop:deviation}, the deviation bound (A3) holds with probability at least $1-12c_1\exp[-(c_2^2-1)\log(p_1p_2)]$, where $\mathbb{Q}$ is specified in (\ref{Q-expression}). Combine all sample size requirement, the leading term becomes $n\succsim \log(p_1p_2)$. Therefore, for random pair $(X,E)$, with probability at least 
\begin{equation*}
1-12c_1\exp[-(c_2^2-1)\log (p_1p_2)] -  2\exp(-c_3n),
\end{equation*}
for some $c_1>0,c_2>1,c_3>0$, the bound in (\ref{bbound-1}) holds, as the result of Theorem~\ref{thm:ErrorBound_beta} and Proposition~\ref{prop:REcondition} and \ref{prop:deviation} combined. 
\end{proof}

\medskip

\begin{proof}[Proof of Proposition~\ref{prop:residual-concentration}]
First we note the following decomposition:
\begin{equation*}
\|\widehat{S} - \Sigma^*_\epsilon\|_\infty \leq   \| S - \Sigma_\epsilon \|_\infty + \|\widehat{S} - S \|_\infty : = \|W_1\|_\infty + \|W_2\|_\infty
\end{equation*}
where $S$ is the sample covariance matrix of the true errors $E$. 

For $W_1$, by Lemma~8 in \citet{ravikumar2011high}, for sample size $$n\geq 512(1+4\sigma_\epsilon^2)^4 \max_i(\Sigma^*_{\epsilon,ii})^4\log(4p_2^{\tau_2}),$$ the following bound holds w.p. at least $1-1/p_2^{\tau_2-2}(\tau_2>2)$:
\begin{equation}\label{exp:c_epsilon}
\|W_1\|_\infty \leq \sqrt{\frac{\log 4 + \tau_2 \log p_2}{c^*_\epsilon n}}, \quad \text{where }c_\epsilon^*=\left[ 128(1+4\sigma_\epsilon^2)^2\max\limits_{i}(\Sigma_{\epsilon,ii}^*)^2\right]^{-1}.
\end{equation}
For $W_2$, re-write it as:
\begin{equation}\label{W2:decom}
W_2 = \frac{2}{n} E'X(B^*-\widehat{B}) + (B^*-\widehat{B})'\left(\frac{X'X}{n}\right)  (B^*-\widehat{B}) 
\end{equation}
The first term in (\ref{W2:decom}) can be bounded as:
\begin{equation}\label{W1bound}
\left\|\frac{2}{n} E'X(B^*-\widehat{B})\right\|_\infty \leq 2 \vertiii{B^*-\widehat{B}}_1\left\|\frac{1}{n}X'E\right\|_\infty \leq 2\|\beta^*-\widehat{\beta}\|_1\cdot \left\|\frac{1}{n}X'E\right\|_\infty.
\end{equation}
By Lemma~\ref{lemma:deviation_aux},  with probability at least $1-6c_1\exp[-(c_2^2-1)\log (p_1p_2)]$, the following bound holds:
\begin{equation}\label{W2bound:1st}
\left\|\frac{2}{n} E'X(B^*-\widehat{B})\right\|_\infty \leq 2c_2\nu_\beta \left[\Lambda_{\max}(\Sigma_X)\Lambda_{\max}(\Sigma^*_\epsilon)\right]^{1/2} \sqrt{\frac{\log (p_1p_2) }{n}},
\end{equation}
with the sample size requirement being $n\succsim \log(p_1p_2)$. 

For the second term in (\ref{W2:decom}), we consider the following bound:
\begin{equation} \label{W2bound:2nd}
\begin{split}
\|(B^*-\widehat{B})'\left(\frac{X'X}{n}\right)  (B^*-\widehat{B}) \|_\infty & \leq \vertiii{B^*-\widehat{B}}_1 \left\|\left(\frac{X'X}{n}\right) (B^*-\widehat{B})\right\|_\infty\\
& \leq  \vertiii{B^*-\widehat{B}}^2_1   \left\| \left(\frac{X'X}{n}\right) \right\|_\infty
\end{split} 
\end{equation}
Here, we apply Lemma~8 in \citet{ravikumar2011high} to the design matrix $X$, for sample size $$n\geq 512(1+4\sigma_x^2)^4 \max_i(\Sigma_{X,ii})^4\log (4p_1^{\tau_1}),$$ the following bound holds w.p. at least $1-1/p_1^{\tau_1-2}(\tau_1>2)$:
\begin{equation}\label{exp:c_X}
\left\|\left(\frac{X'X}{n}\right) - \Sigma_X\right\|_\infty \leq \sqrt{\frac{\log 4 + \tau_1 \log p_1}{c^{*}_X n}}, \qquad \text{where }~c_X^*=\left[ 128(1+4\sigma_x^2)^2\max\limits_{i}(\Sigma_{X,ii})^2\right]^{-1}
\end{equation}
This indicates that with this choice of $n$, the following bound holds with probability at least $1-1/p_1^{\tau_1-2}(\tau_1>2)$:
\begin{equation*}
\left\| \left(\frac{X'X}{n}\right) \right\|_\infty \leq \sqrt{\frac{\log 4 + \tau_1 \log p_1}{c^*_Xn}} +  \max_i(\Sigma_{X,ii})
\end{equation*}
Combine with the bound in (\ref{W2bound:2nd}), with probability at least $1-1/p_1^{\tau_1-2}(\tau_1>2)$, the following bound holds:
\begin{equation} \label{W2bound:2nd-final}
\|(B^*-\widehat{B})'\left(\frac{X'X}{n}\right)  (B^*-\widehat{B}) \|_\infty  \leq \nu_\beta^2\left( \sqrt{\frac{\log 4 + \tau_1 \log p_1}{c^*_Xn}} +  \max_i(\Sigma_{X,ii})\right)
\end{equation}
Now combine (\ref{W1bound}), (\ref{W2bound:1st}) and (\ref{W2bound:2nd-final}), we reach the conclusion of Proposition 3, with the leading term in the sample size requirement being $n\succsim \log(p_1p_2)$. 
\end{proof}

\subsection{An example for multi-layered network estimation.} \label{appendix:example}
As mentioned at the beginning of Section~\ref{sec:methodology}, the proposed methodology is designed for obtaining MLEs for multi-layer Gaussian networks, but the problem breaks down into a sequence of 2-layered estimation problems. Here we give an detailed example to illustrate how our proposed methodology proceeds for a 3-layered network. \\

Suppose there are $p_1,p_2$ and $p_3$ nodes in Layers 1, 2 and 3, respectively. This three-layered network is modeled as follows: 
\begin{itemize}
 \setlength\itemsep{1pt}
\renewcommand\labelitemi{--}
\item $\bs{X}\sim \mathcal{N}(0,\Sigma_X)$, $\bs{X}\in\mathbb{R}^{p_1}$. 
\item For $j=1,\cdots,p_2$: $Y_j = \bs{X}'B_j^{xy}+\epsilon^Y_j$, $B_j^{xy}\in\mathbb{R}^{p_1}$. $(\epsilon^Y_1  \cdots \epsilon^Y_{p_2})'\sim \mathcal{N}(0,\Sigma_{\epsilon,Y})$. 
\item For $l=1,2,\cdots,p_3$: $Z_l = \bs{X}'B_l^{xz} + \bs{Y}'B_l^{yz}+\epsilon_l^Z$, $B_l^{xz}\in\mathbb{R}^{p_1}$ and $B_l^{yz}\in\mathbb{R}^{p_2}$. 
$(\epsilon^Z_1 \cdots \epsilon^Z_{p_3})'\sim \mathcal{N}(0,\Sigma_{\epsilon,Z})$.
\end{itemize}
The parameters of interest are : $\Theta_X$, $\Theta_{\epsilon,Y}:=\Sigma_{\epsilon,Y}^{-1}$, $\Theta_{\epsilon,Z}:=\Sigma_{\epsilon,Z}^{-1}$, which denote the within-layer conditional dependencies, and
\begin{equation*}
B_{XY} = \begin{bmatrix}
B_1^{xy} & \cdots & B_{p_2}^{xy}
\end{bmatrix},~~  B_{XZ} = \begin{bmatrix}
B_1^{xz} & \cdots & B_{p_3}^{xz}
\end{bmatrix} ~~\text{and}~~ B_{YZ} = \begin{bmatrix}
B_1^{yz} & \cdots & B_{p_3}^{yz},
\end{bmatrix}
\end{equation*}
which encode the across-layer dependencies.\\

Now given data $X\in\mathbb{R}^{n\times p_1}$, $Y\in\mathbb{R}^{n\times p_2}$ and $Z\in\mathbb{R}^{n \times p_3}$, all centered, the full log-likelihood can be written as:
\begin{equation}
\ell(Z,Y,X) = \ell(Z|Y,X; \Theta_{\epsilon,Z},B_{YZ},B_{XZ}) + \ell(Y|X;\Theta_{\epsilon,Y},B_{XY}) + \ell(X;\Theta_X).
\end{equation}
The separability of the log-likelihood enables us to ignore the inner structure of the combined layer $\widetilde{X}:=(X,Y)$ when trying to estimate the dependencies between Layer 1 and Layer 3, Layer 2 and Layer 3, as well as the conditional dependencies within Layer 3. As a consequence, the optimization problem minimizing the negative log-likelihood can be decomposed into three separate problems, i.e., solving for $\{\Theta_{\epsilon,Z},B_{XZ},B_{YZ}\}$, $\{\Theta_{\epsilon,Y},B_{XY}\}$ and $\{\Theta_X\}$, respectively. 

The estimation procedure described in Section~\ref{sec:estimation} can thus be carried out in a recursive way in a sense of what follows. To obtain estimates for $\{B_{XZ},B_{YZ},\Theta_{\epsilon,Z}\}$, based on the formulation in (\ref{eqn:obj}), we solve the following opmization problem:
\begin{eqnarray*}
\min\limits_{\substack{\Theta_{\epsilon,Z}\in\mathbb{S}_{++}^{p_3\times p_3}\\B_{XZ},B_{YZ}}}  \begin{Bmatrix} -\log\det\Theta_{\epsilon,Z} + \frac{1}{n}\sum_{j=1}^{p_3}\sum_{i=1}^{p_3} \sigma_Z^{ij}(Z_i-XB^{xz}_i-YB^{yz}_i)^\top(Z_j-XB^{xz}_j-YB^{yz}_j)  \vspace*{3mm}\\
+\lambda_n(\|B_{XZ}\|_1 + \|B_{YZ}\|_1)+ \rho_n\|\Theta_{\epsilon,Z}\|_{1,\text{off}}  \end{Bmatrix},
\end{eqnarray*}
which can be solved by treating the combined design matrix $\widetilde{X}=(X,Y)$ as a single super layer and $Z$ as the response layer, then apply each step described in Section~\ref{sec:estimation}. To obtain estimates for $B_{XY}$ and $\Theta_{\epsilon,Y}$, we can ignore the 3rd layer for now and apply the exact procedure all over again, by treating $Y$ as the response layer and $X$ as the design layer. The estimate for the precision matrix of the bottom layer $\Theta_X$ can be obtained by graphical lasso \citep{friedman2008sparse} or the nodewise regression \citep{meinshausen2006high}. \\

\end{document}